\documentclass[journal,onecolumn,11pt]{article}

\usepackage[a4paper,margin=25mm]{geometry}
\usepackage{amsmath,amssymb,amsthm,booktabs,longtable,array,enumitem}
\usepackage{float}
\usepackage{xcolor}
\usepackage{graphicx}
\usepackage[numbers,sort&compress]{natbib}
\usepackage{algorithm,algpseudocode}
\usepackage[colorlinks=true,linkcolor=blue,citecolor=blue,urlcolor=blue]{hyperref}
\newtheorem{theorem}{Theorem}
\newtheorem{proposition}{Proposition}
\newtheorem{lemma}{Lemma}

\theoremstyle{definition}

\newtheorem{remark}{Remark}
\newtheorem{example}{Example}

\newcommand{\Z}{\mathbb Z}
\newcommand{\F}{\mathbb F_2}
\newcommand{\RM}{\operatorname{RM}}
\newcommand{\rk}{\operatorname{rank}}
\newcommand{\Ker}{\operatorname{K}}

\newcommand{\dig}{\operatorname{dig}}
\renewcommand{\vec}[1]{\boldsymbol{#1}}   
\newcommand{\spn}{\operatorname{span}}

\begin{document}
	
\title{The rank of $\mathbb{Z}_{2^s}$-linear Hadamard codes
	\thanks{This research (M. Shi) is supported by the National Natural Science Foundation of China under Grants W2611002 and 12471490.}}
\author{Xuan Wang
	\thanks{Xuan Wang is with the School of Mathematical Sciences, Anhui University, Hefei 230601, China (email: wang\_xuan\_ah@163.com)} \and  
	Minjia Shi\thanks{Minjia Shi is with the Key Laboratory of Intelligent Computing Signal Processing, Ministry of Education, School of Mathematical Sciences, Anhui University, Hefei 230601, China. (email: smjwcl.good@163.com).}
}

	\maketitle
	
	\begin{abstract}
		For every $s\ge1$, we express the binary rank of a canonical $\Z_{2^s}$-linear Hadamard code as a binomial polynomial in its type parameters. Its coefficients $\rho_s(\vec k)$ are dimensions of exact-support spaces of Boolean digit functions and depend only on $s$ and the multiplicity vector $\vec k$. We give a finite exact algorithm to compute these coefficients once for each fixed $s$. We determine all coefficients for $s\le5$, recovering the known formulas over $\Z_4$ and $\Z_8$ and obtaining explicit binomial polynomials over $\Z_{16}$ and $\Z_{32}$. For general $s$, we evaluate the coefficients at and immediately below the maximal support weight and give a closed formula for a single factor of offset zero.
	\end{abstract}
	

	\section{Introduction}\label{sec:intro}
	
	Let $s\ge1$ and let $\Z_{2^s}$ be the ring of integers modulo $2^s$.  The set of $n$-tuples over $\Z_{2^s}$ is denoted by $\Z_{2^s}^n$, and its elements are called vectors.  The \emph{order} $o(\vec u)$ of $\vec u\in\Z_{2^s}^n$ is the least positive integer $m$ with $m\vec u=\vec 0$.  A \emph{$\Z_{2^s}$-additive code} is a subgroup $\mathcal C\subseteq\Z_{2^s}^n$. It is of \emph{type} $(n;t_1,\dots,t_s)$ when
	\[
	\mathcal C\cong\Z_{2^s}^{t_1}\times\Z_{2^{s-1}}^{t_2}\times\cdots\times\Z_2^{t_s}.
	\]
	For $N \geq 2$, a binary code $C\subseteq\F^N$ is called a \emph{Hadamard code} if $|C|=2N$ and its minimum distance is $N/2$.  
	Such a code consists of the rows of a Hadamard matrix together with their complements.  
	Hadamard codes are closely related to designs \cite{AssmusKey,MS}.  The linear Hadamard codes are exactly the first-order Reed--Muller codes $\RM(1,t)$, that is, the duals of the extended Hamming codes \cite{MS}.  Most Hadamard codes are nonlinear, and their classification is a classical problem.  Algebraic techniques for nonlinear codes are developed in \cite{Bauer}.  A fruitful approach, initiated by the $\Z_4$-linear representation of the Kerdock, Preparata and Goethals codes \cite{Hammons}, is to realise a nonlinear Hadamard code as the Gray image of an additive code over a ring $\Z_{2^s}$ or over a mixed alphabet, and then to decide which of the resulting codes are equivalent by means of invariants.
	
	Carlet's Gray map $\phi_s:\Z_{2^s}\to\F^{2^{s-1}}$ is defined by
	\[
	\phi_s(u)=(u_{s-1},\dots,u_{s-1})+(u_0,\dots,u_{s-2})Y_{s-1},
	\]
	where $u=\sum_{j=0}^{s-1}u_j2^j$ with $u_j\in\{0,1\}$ and the columns of the $(s-1)\times2^{s-1}$ binary matrix $Y_{s-1}$ are all the vectors of $\F^{s-1}$. Thus $\phi_1$ is the identity and $\phi_2$ is the classical Gray map.  We write $\Phi_s=\phi_s^{\times n}$ for the componentwise extension.  If $\mathcal C\subseteq\Z_{2^s}^n$ is additive, then $\Phi_s(\mathcal C)$ is called a \emph{$\Z_{2^s}$-linear code}, of length $n2^{s-1}$.  For every type $\vec a=(a_1,\dots,a_s)$ with $a_1\ge1$, the recursive construction of \cite{DCC2022} produces a canonical $\Z_{2^s}$-additive code $\mathcal H^{\vec a}$ whose Gray image $H^{\vec a}$ is a binary code of length $2^t$, where 
	\begin{equation}\label{eq:length}
		t+1=\sum_{i=1}^{s}(s-i+1)a_i .
	\end{equation}
	For $t\ge1$, $H^{\vec a}$ is a binary Hadamard code. 
	When $s=1$, the canonical binary code $H^{a_1}$ has rank $a_1$.
	For $a_1=1$, this construction gives the length-one code $H^1=\{0,1\}$, whose minimum distance is $1$. We include this degenerate case in the rank formula. 
	For $s=2$, the codes $H^{\vec a}$ are the $\Z_4$-linear Hadamard codes. 
	Carlet's Gray map provides the $\Z_{2^s}$-linear framework \cite{Carlet}.  
	Codes over integer residue rings go back to Blake \cite{Blake} and Shankar \cite{Shankar}, and linear codes over $\Z_{2^s}$ were further studied in \cite{Gupta}.  
	Several Gray maps and self-dual codes over $\Z_{2^k}$ are studied in \cite{Hammons,Dougherty}.  
	Every $\Z_{2^k}$-code is a binary propelinear code \cite{FCVRB2001}.  
	For every $s$, every $\Z_{2^s}$-linear Hadamard code is permutation equivalent to one of the canonical codes $H^{\vec a}$ \cite{CFV2019}.  
	The classification problem for the family is therefore reduced to deciding when two types $\vec a$ and $\vec a'$ give equivalent codes, and the two classical invariants for this comparison are the rank and the kernel.
	
	The \emph{rank} of a binary code $C$ is the dimension $\rk(C)=\dim\langle C\rangle$ of its linear span, and its \emph{kernel} is
	\[
	\Ker(C)=\{\vec x:\vec x+C=C\}.
	\]
	We write $\ker(C)=\dim\Ker(C)$ for the dimension of the kernel.  
	Since linear Hadamard codes are the first-order Reed--Muller codes, their rank is $t+1$ and their kernel is the whole code. 
	For a nonlinear Hadamard code normalized so that $\vec0\in C$, the rank is strictly larger than the linear value $t+1$, while the kernel is a proper subcode, so both invariants carry information.  
	If $\vec 0\in C$ then $\Ker(C)$ is a linear subcode of $C$, and if $C$ is linear then $\Ker(C)=C=\langle C\rangle$.  
	Both invariants are preserved by coordinate permutations. The kernel is unchanged by translation. 
	If $0\in C$ and $0\in C+\vec v$, then $\vec v\in C\cap(C+\vec v)$, so $C$ and $C+\vec v$ span the same binary space. Hence normalized codes with different pairs $(\rk,\ker)$ are nonequivalent. 
	The converse fails in general, and the failure is precisely what makes the classification of nonlinear Hadamard codes difficult. A substantial part of the literature is devoted to deciding, family by family, how far the pair $(\rk,\ker)$ goes and which finer invariants are needed when it does not suffice.
	
	For $\Z_4$-linear Hadamard codes either the rank or the kernel dimension classifies the family completely \cite{Krotov2001,PRV2006}.  
	There are exactly $\lfloor(t-1)/2\rfloor$ nonequivalent such codes of length $2^t$ for $t\ge3$ \cite{CFV2019}.  For $\Z_8$-linear Hadamard codes the rank and the kernel dimension together give a complete classification, and the rank admits a closed formula \cite{Z8}.  
	For the mixed $\Z_2\Z_4$-linear and $\Z_2\Z_4\Z_8$-linear families the situation is more recent.  The $\Z_2\Z_4$-linear Hadamard codes were classified in \cite{KrotovV2015}, and generator matrices and duality for $\Z_2\Z_4$-linear codes are given in \cite{Borges}.  
	The rank of every $\Z_2\Z_4\Z_8$-linear Hadamard code is now known in closed form, and the rank together with the kernel dimension classifies the family except for two infinite pairs of types \cite{BhuniaR8}.  
	
	For each nonlinear $\Z_2\Z_4$-linear or $\Z_{2^s}$-linear Hadamard code, the local rank is the same on all of its kernel blocks.
	For the $\Z_2\Z_4\Z_8$-linear family, the kernel-block rank profile and its two-block refinement are used to complete the classification in \cite{Profile}.  
	In the pure $\Z_{2^s}$ family, the linearity and the dimension of the kernel were determined in \cite{CFV2019}, and the equivalence classification was completed only recently \cite{Bhunia2026}, after the partial results of \cite{Equiv,CFV2019}: for $t\ge3$ there are $\mathcal P(t+1)-\mathcal P(t)-1$ nonequivalent $\Z_{2^s}$-linear Hadamard codes of length $2^t$, where $\mathcal P$ is the partition function and $s\ge2$ is allowed to vary.  
	For odd primes the linearity and the kernel are treated in \cite{DCC2022}, the classification of the generalized Hadamard codes is completed in \cite{Bhunia2026}, and the relation between the $\Z_p\Z_{p^2}$-linear and the $\Z_{p^s}$-linear families is studied in \cite{BhuniaOdd}. 
	For general $\Z_p\Z_{p^2}$-linear codes, see \cite{Li2024Rank}. 
	
	The present paper is concerned with the remaining numerical invariant, the rank.  
	For $s\le3$ it is known in closed form: for $\Z_4$,
	\[
	\rk H^{a_1,a_2}=2a_1+a_2+\binom{a_1-1}{2},
	\]
	and for $\Z_8$ the quartic formula of \cite[Theorem 4.1]{Z8}.  
	For $s\ge4$, the cited work gives ranks for particular types and lengths. 
	For example, the $\Z_{16}$ ranks of the admissible types of length $2^9$, $2^{10}$ and $2^{11}$ are tabulated in \cite[Table 3]{CFV2019}, the $\Z_8$-linear codes of length $2^t$ are tabulated for $t\le17$ in \cite{Z8}, and \cite[Section 7]{Profile} and \cite{Bhunia2026} recall that in the other cases the rank was computed with Magma \cite{Bosma} for small lengths.  The recent closed formula for the mixed $\Z_2\Z_4\Z_8$ family \cite{BhuniaR8} does not cover the pure $\Z_{2^s}$ family with $s\ge4$.  What is missing is a method that produces the rank for every $s$ without treating the types and lengths one at a time.
	
	The main result of this paper is such a method.  We show that, for every $s$ and every type $\vec a=(a_1,\dots,a_s)$,
	\[
	\rk H^{\vec a}=\sum_{\vec k}\rho_s(\vec k)\prod_{\ell=0}^{s-1}\binom{\mu_\ell}{k_\ell},
	\qquad \mu_0=a_1-1,\quad \mu_\ell=a_{\ell+1}\ (1\le\ell<s),
	\]
	where the coefficients $\rho_s(\vec k)$ depend only on the alphabet exponent $s$ and on the multiplicity vector $\vec k=(k_0,\dots,k_{s-1})$, and neither on the type nor on the length.  The sum is finite, over the feasible vectors with $\sum_\ell2^\ell k_\ell\le2^{s-1}$. Their number is
	\[
	N_s=\bigl|\{\vec k:\textstyle\sum_\ell2^\ell k_\ell\le2^{s-1}\}\bigr|=4,10,36,202,\dots
	\]
	for $s=2,3,4,5,\dots$. 
	
	The paper is organized as follows.  
	Section~\ref{sec:notation} fixes the concepts and the notation.  
	Section~\ref{sec:rank} presents the rank formula in terms of $\rho_s(\vec{k})$ and the binomial coefficients. 
	Section \ref{sec:rho} contains the values of $\rho_s(\vec{k})$ for several special $s$ and $\vec{k}$, and deals with the ranks over $\Z_4$, $\Z_8$, $\Z_{16}$ and $\Z_{32}$, tabulates the coefficients $\rho_s(\vec{k})$ for $s\le5$.   
	Section \ref{sec:app} discusses the applications of rank, and shows how a sequence of ranks distinguishes codes. 
	Section~\ref{sec:conclusion} concludes the paper.  
	Appendix~\ref{app:algo} and Appendix~\ref{app:rho} give the pseudocode and the theory to compute $\rho_s(\vec{k})$, respectively.

	\section{Concepts and notations}\label{sec:notation}
	
	\subsection{Carlet's Gray map}
	For $u=\sum_{j=0}^{s-1}u_j2^j\in\Z_{2^s}$, Carlet's Gray map $\phi_s:\Z_{2^s}\to\F^{2^{s-1}}$ is
	\begin{equation}\label{eq:gray}
		\phi_s(u)=(u_{s-1},\dots,u_{s-1})+(u_0,\dots,u_{s-2})Y_{s-1},
	\end{equation}
	where the columns of the $(s-1)\times2^{s-1}$ binary matrix $Y_{s-1}$ are all the vectors of $\F^{s-1}$. Thus $\phi_1$ is the identity and $\phi_2$ is the classical Gray map.  
	We put $\Phi_s=\phi_s^{\times n}$ for the componentwise extension.  
	We use the following \emph{additivity} of $\phi_s$, which is immediate from \eqref{eq:gray}:
	\[
	\phi_s(2^{s-1}\lambda+w)=\phi_s(2^{s-1}\lambda)+\phi_s(w)\qquad(\lambda,w\in\Z_{2^s}),
	\]
	and the digit identity
	\[
	2^{s-1}(w\bmod 2)=2^{s-1}w\quad\text{in }\Z_{2^s},
	\]
	because the two sides differ by $2^{s-1}$ times an even integer.
	
	\subsection{\texorpdfstring{$\Z_{2^s}$}{}-Hadamard codes}
	Let $s\ge1$ and let $\Z_{2^s}$ be the ring of integers modulo $2^s$. The order $o(u)$ of $u\in\Z_{2^s}$ is the least $m\ge1$ with $mu=0$.  
	A \emph{$\Z_{2^s}$-additive code} is a subgroup $\mathcal C\subseteq\Z_{2^s}^n$. It is of \emph{type} $(n;t_1,\dots,t_s)$ when $\mathcal C\cong \Z_{2^s}^{t_1}\times\Z_{2^{s-1}}^{t_2}\times\cdots\times\Z_2^{t_s}$.  
	For $1\le i\le s$ put
	\[
	T_i=\{j2^{i-1}:0\le j\le 2^{s-i+1}-1\}=2^{i-1}\Z_{2^s},
	\]
	an ideal of $\Z_{2^s}$ of exponent $2^{s-i+1}$. A nonzero element of $T_i$ may have a smaller order (for instance $2\in T_1\subseteq\Z_8$ has order $4$), while a row ranging over the whole ideal has order $2^{s-i+1}$.  
	A row of a generator matrix of $\mathcal C$ whose entries run over $T_i$ has \emph{level} $i$ and order $2^{s-i+1}$. The \emph{order} of an element $w$ of $\Z_{2^s}^n$ is the least $m\ge1$ with $mw=\vec 0$. 
	Thus a row of level $i$ has order $2^{s-i+1}$.
	
	Let $a_1,\dots,a_s$ be nonnegative with $a_1\ge1$ and $m=a_1+\cdots+a_s$.  
	The canonical matrix $A^{\vec a}=A^{a_1,\dots,a_s}$ has the $m$ rows $w_1,\dots,w_m$, with $w_1$ the all-one row, and as columns exactly all the vectors $\vec z^{T}$ with
	\[
	\vec z=(1,z_1,\dots,z_{m-1})\in \mathcal{G}^{\vec{a}} := \{1\}\times T_1^{a_1-1}\times T_2^{a_2}\times\cdots\times T_s^{a_s}.
	\]
	the entry $z_i$ is placed in the row $w_{i+1}$.  For $1\le i\le s$ exactly $a_i$ of the rows have level $i$, and the rows are ordered by increasing level. We let $q_i+1$ be the level of the row $w_{i+1}$, so that $q_i\in\{0,\dots,s-1\}$ is the \emph{valuation offset} of the factor $i$, and then $z_i\in T_{q_i+1}=2^{q_i}\Z_{2^s}$.  
	The \emph{normalized coordinate} of the factor $i$ in the column $\vec z$ is
	\begin{equation}\label{eq:zeta}
		\zeta_i=\frac{z_i}{2^{q_i}}\in\Z_{2^{s-q_i}},\qquad 1\le i<m.
	\end{equation}
	Thus a column is described equivalently by $(1,z_1,\dots,z_{m-1})$ or by its normalized coordinates $(\zeta_1,\dots,\zeta_{m-1})$. 
	
	We write $\mathcal G$ for $\mathcal G^{\vec a}$ when the type is clear, and $\mathcal{H}^{a_1,\dots,a_s}$ for the $\Z_{2^s}$-additive code generated by the rows of $A^{a_1,\dots,a_s}$ and $H^{a_1,\dots,a_s}=\Phi_s(\mathcal{H}^{a_1,\dots,a_s})$ for its Gray image. 
	By \cite{CFV2019}, the binary code $H^{a_1,\dots,a_s}$ has length $2^t$, where $t$ is given by \eqref{eq:length}. It is a Hadamard code for $t\ge 1$. 
	Every binary $\Z_{2^s}$-linear Hadamard code is permutation equivalent to a canonical Gray image $H^{\vec a}$ for some type $\vec a$ \cite{CFV2019}.
	
	\begin{remark}
		The $\Z_{2^s}$-linear Hadamard code $H^{a_1,\dots,a_s}$ is called \textbf{reduced} if $a_1 \geq 2$. 
		Let $H^{\vec{a}}$ and $H^{\vec{a}'}$ be two reduced codes of the same length. 
		Then by \cite[Theorem 7.2]{Bhunia2026}, $H^{\vec{a}}$ and $H^{\vec{a}'}$ are equivalent if and only if $\vec{a} = \vec{a}'$. 
	\end{remark}
	
	\begin{remark}[\cite{CFV2019,Krotov2001,Gupta,PRV2006}] \label{rem:linear}
		For $s\ge2$, the binary code $H^{a_1,\dots,a_s}$ is linear over $\mathbb{F}_2$ if and only if one of the following conditions hold: 
		\begin{enumerate}[label=$(\arabic*)$]
			\item $a_1 \leq 2$ and $s=2$; 
			\item $a_1=1$, $a_2 = \dots = a_{s-2} = 0$, $a_{s-1} = 1$, $s \geq 3$; 
			\item $a_1=1$, $a_2 = \dots = a_{s-1} = 0$, $s \geq 3$. 
		\end{enumerate}
	\end{remark}
	
	\subsection{The columns of the binary Gray image \texorpdfstring{$H^{a_1,\dots,a_s}$}{}}
	
	It is known that the code $\mathcal{H}^{a_1,\dots,a_s}$ is generated by the rows of $A^{a_1,\dots,a_s}$. 
	Then every codeword of $\mathcal{H}^{a_1,\dots,a_s}$ can be written as $\vec{x} A^{a_1,\dots,a_s}$, where $\vec{x} \in \mathcal{M}$ is called a \textbf{message}. 
	The representation is unique when each row coefficient is taken modulo the order of that row, as specified by the message space below. 
	This does not assert linear independence over $\Z_{2^s}$: a row of order less than $2^s$ has a nonzero annihilator. 
	The message space is 
	\[ 
	\mathcal{M} = \mathbb{Z}_{2^s}^{a_1} \times \mathbb{Z}_{2^{s-1}}^{a_2} \times \cdots \times \mathbb{Z}_{2}^{a_s}. 
	\]
	In particular, for the message $\vec{x} = (b,x_1,\dots,x_{m-1})$, $\mathcal{M}$ can also be written as 
	\[ \mathcal{M} = \mathbb{Z}_{2^s} \times \mathbb{Z}_{2^{s-q_1}} \times \cdots \times \mathbb{Z}_{2^{s-q_{m-1}}}, \]
	where $b\in\Z_{2^s}$ and, for $1\le i<m$, $x_i\in\mathbb{Z}_{2^{s-q_i}}$ is the coefficient of a factor of level $q_i+1$. 
	\begin{remark}
		To prove uniqueness, subtract two messages giving the same codeword. At the canonical column with all $z_i=0$, the resulting zero codeword has entry $b$, so $b=0$. For each $i$, the canonical column with $z_i=2^{q_i}$ and all other nonconstant coordinates zero then gives
		\[
		2^{q_i}x_i=0\pmod{2^s},\qquad
		x_i=0\pmod{2^{s-q_i}}.
		\]
		Thus the difference is zero in the stated message space, proving injectivity and hence uniqueness.
	\end{remark}
	
	For the message $\vec x=(b,x_1,\dots,x_{m-1})$, we let $\vec{u}=\vec{x} A^{a_1,\dots,a_s}$ be the corresponding codeword of $\mathcal{H}^{a_1,\dots,a_s}$ and $u(\vec z)$ its entry at the column $\vec z$, which is 
	\begin{equation}\label{eq:W}
		W_{\vec x}(\vec{z}) = u(\vec z) = b+\sum_{i=1}^{m-1} x_i z_i \in\Z_{2^s},
	\end{equation}
	where $b$ is the coefficient of the all-one row $w_1$ and $x_i$ the coefficient of the row $w_{i+1}$.  
	So, $\vec{u} = (W_{\vec x}(\vec{z}))_{\vec{z} \in \mathcal{G}}$ is a row of the code matrix $\mathcal{H}^{a_1,\dots,a_s}$, whose rows are indexed by messages.  
	Let $\dig_j(u)$ be the binary $j$-th digit of $u$, i.e., 
	\[ u = \sum_{j=0}^{s-1}\dig_j(u)2^j\qquad (0\le u<2^s).  \]
	For example, 
	\[ W_{\vec{x}}(\vec{z}) = W_0 \cdot 2^0 + W_1 \cdot 2^1 + \dots + W_{s-1} \cdot 2^{s-1},  \]
	where $W_j = \dig_j(W_{\vec{x}}(\vec{z}))$ and $W_j = 0$ when $j\ge s$, using the standard representative modulo $2^s$. 
	Hence, by \eqref{eq:gray}, the Gray image of $W_{\vec x}(\vec z)$ is
	\[ \phi_s\bigl(W_{\vec x}(\vec z)\bigr) = (W_{s-1},\dots,W_{s-1})+(W_0,\dots,W_{s-2}) Y_{s-1}.  \]
	Let $\vec{h}=(h_0,\dots,h_{s-2})\in\F^{s-1}$ be a column of $Y_{s-1}$. 
	Then the coordinate with index $\vec{h}$ of $\phi_s\bigl(W_{\vec x}(\vec z)\bigr)$ is 
	\begin{equation}\label{eq:coord}
		\phi_s\bigl(W_{\vec x}(\vec z)\bigr)_{\vec h} = \dig_{s-1}\bigl(W_{\vec x}(\vec z)\bigr) + \sum_{j=0}^{s-2} h_j\, \dig_j\bigl(W_{\vec x}(\vec z)\bigr),
	\end{equation}
	since $W_j = \dig_j(W_{\vec{x}}(\vec{z}))$. 
	The above idea is from the proof of \cite[Theorem 5.2]{Profile}. 
	
	\begin{remark}\label{rem:column}
		Summarizing the discussions above, we can say, the columns of the binary image $H^{a_1,\dots,a_s}$ can be indexed by the pairs $(\vec{z}, \vec{h})$. 
		And the column corresponding to $(\vec{z}, \vec{h})$ is exactly 
		\[ \left(\phi_s\bigl(W_{\vec x}(\vec z)\bigr)_{\vec h}\right)_{\vec{x} \in \mathcal{M}}. \]
	\end{remark}
	
	\subsection{Kernel and Rank}
	
	The \emph{rank} of a binary code $C$ is the dimension $\rk(C)=\dim\langle C\rangle$ of its linear span. 
	The \emph{kernel} is
	\[
	\Ker(C)=\{\vec x:\vec x+C=C\}.
	\]
	Both invariants are preserved by coordinate permutations. The kernel is unchanged by translation. If $0\in C$ and $0\in C+\vec v$, then $\vec v\in C\cap(C+\vec v)$, so $C$ and $C+\vec v$ span the same binary space. Hence normalized codes with different pairs $(\rk,\ker)$ are nonequivalent. 
	For example, the kernel and rank of $\Z_4$-linear Hadamard codes \cite{Krotov2001,PRV2006} and $\Z_8$-linear Hadamard codes \cite{Z8} are determined completely, and the equivalence is also determined. 
	The linearity and the dimension of the kernel were determined in \cite{CFV2019}, and the equivalence classification was completed only recently \cite{Bhunia2026}. 
	
	For the pure $\Z_{2^s}$ family, the closed formulas over $\Z_4$ and $\Z_8$ provide benchmarks for the binomial formula in Section~\ref{sec:rank}. 
	
	The following table collects the symbols used throughout the paper.
	
	\begin{longtable}{@{}p{34mm}p{121mm}@{}}
		\caption{Notations and symbols} \\
		\toprule Symbol & Meaning\\\midrule\endhead
		$s,\ \Z_{2^s}, \ T_i$ & alphabet exponent and ring, $T_i=2^{i-1}\Z_{2^s}$ the ideal of level $i$\\
		$\phi_s,\Phi_s$ & Carlet's Gray map and its extension \\
		$m$, $\vec a=(a_1,\dots,a_s)$ & type: $a_i$ rows of level $i$, $m=\sum_ia_i$, $a_1\ge1$\\
		$A^{\vec{a}}$, or $A^{a_1, \dots, a_s}$ & the generator matrix \\
		$\mathcal{H}^{\vec{a}}$, or $\mathcal{H}^{a_1, \dots, a_s}$ & the $\mathbb{Z}_{2^s}$-code generated by $A^{\vec{a}}$ \\
		$H^{\vec{a}}$, or $H^{a_1, \dots, a_s}$ & the Gray image of $\mathcal{H}^{\vec{a}}$ \\
		$t$ & binary length $2^t$ of $H^{\vec a}$, see \eqref{eq:length}\\
		$(\vec{z},\vec{h})$ & $\vec z =(1,z_1,\dots,z_{m-1})\in\mathcal G^{\vec a}$ a column of $A^{\vec a}$, and $\vec h\in\F^{s-1}$ the column of $Y_{s-1}$ \\
		$q_i,\ \zeta_i$ & level $q_i+1$ of the row $w_{i+1}$, $z_i\in T_{q_i+1}$, normalized $\zeta_i=z_i/2^{q_i}$, see \eqref{eq:zeta}\\
		$\mathcal{M}$ & the set of messages \\
		$(b, x_1,\dots,x_{m-1})$ & a message: $b$ for $w_1$, $x_i$ for $w_{i+1}$ \\
		$\dig_j$ & $j$-th binary digit of an element \\
		$a_{i,r}$, $b_r$ & $\dig_r(x_i)$, $\dig_r(b)$, the $r$-th bit of the coefficient $x_i$ or $b$, $0\le r<s-q_i$\\
		$\mathcal{B}$ & the algebra of Boolean functions of the message bits $a_{i,r}$, $b_r$ \\
		$\rk$ & rank, the dimension of the span of $H^{\vec{a}}$ \\
		$\vec e_i$, $\vec 1$, $\vec 0$ & standard basis vectors, all-one and zero vectors\\
		\bottomrule
	\end{longtable}
	
	\section{The rank formula}\label{sec:rank}
	
	In this section, we will give the rank formula for a $\mathbb{Z}_{2^s}$-Hadamard code $H^{\vec{a}}$, where $\vec{a} = (a_1, \dots, a_s)$ is the type and $a_1 \geq 1$. 
	
	\subsection{The column space}
	
	Recall two functions defined in \eqref{eq:W} and \eqref{eq:coord}. 
	Let 
	\[
	F_{\vec{z}, \vec{h}} := (F_{\vec x, \vec z, \vec h})_{\vec{x} \in \mathcal{M}}, \qquad F_{\vec x, \vec z, \vec h}=\phi_s\bigl(W_{\vec x}(\vec z)\bigr)_{\vec{h}}, 
	\]
	Then by Remark \ref{rem:column}, $F_{\vec x, \vec z, \vec h}$ is the entry of the code $H^{\vec{a}}$ whose row is indexed by $\vec{x}$ and column is indexed by $(\vec{z},\vec{h})$, and $F_{\vec z, \vec h}$ is exactly the column of $H^{\vec{a}}$ indexed by $(\vec{z}, \vec{h})$. 
	Similarly, the vector 
	\[ W_{\vec{z}} = (W_{\vec x}(\vec z))_{\vec{x} \in \mathcal{M}} \]
	is a column of $\mathcal{H}^{\vec{a}}$ over $\mathbb{Z}_{2^s}$. 
	
	We compute the rank in the function algebra
	\[
	\mathcal B=\{f:\mathcal M\to\F\},
	\]
	with pointwise addition and multiplication. Fix the same ordering of $\mathcal M$ as the ordering of the rows of the code matrix, and define the evaluation map
	\[
	\operatorname{ev}_{\mathcal M}:\mathcal B\longrightarrow\F^{|\mathcal M|},
	\qquad f\longmapsto\bigl(f(\vec x)\bigr)_{\vec x\in\mathcal M}.
	\]
	This is a linear isomorphism: a function is uniquely determined by its values at all messages, and every vector indexed by $\mathcal M$ defines such a function. In particular, evaluation preserves linear dependence and the dimension of every function subspace. 
	
	For $0\le j<s$ and $\vec z\in\mathcal G$, define the message function
	\[
	f_{j,\vec z}:\mathcal M\to\F,\qquad
	f_{j,\vec z}(\vec x)=\dig_j\bigl(W_{\vec x}(\vec z)\bigr).
	\]
	Its evaluation vector is 
	\[
	\operatorname{ev}_{\mathcal M}(f_{j,\vec z}) 
	= \bigl(f_{j,\vec z}(\vec x)\bigr)_{\vec x\in\mathcal M} 
	= \bigl(\dig_j\bigl(W_{\vec x}(\vec z)\bigr)\bigr)_{\vec x\in\mathcal M} 
	= \dig_j(W_{\vec{z}}), 
	\]
	where the digit map on the right is applied componentwise. We distinguish functions from their evaluation vectors throughout the argument.
	
	\begin{lemma}\label{lem:digitspan}
		Let
		\[
		V_s=\spn_{\F}\{f_{j,\vec z}:0\le j<s,\ \vec z\in\mathcal G\}\subseteq\mathcal B.
		\]
		Let $\operatorname{Col}(H^{\vec a})$ denote the column space of the binary matrix whose rows are all the codewords of $H^{\vec a}$, indexed by $\mathcal M$. Then
		\[
		\operatorname{Col}(H^{\vec a})=\operatorname{ev}_{\mathcal M}(V_s),
		\qquad \rk H^{\vec a}=\dim_{\F}V_s.
		\]
	\end{lemma}
	\begin{proof}
		For every pair $(\vec z,\vec h)$ define the function
		\[
		c_{\vec z,\vec h}(\vec x)=F_{\vec x,\vec z,\vec h}.
		\]
		Its evaluation vector is exactly the column $F_{\vec z,\vec h}$ defined above. By \eqref{eq:coord}, the following is an identity of functions on $\mathcal M$:
		\[
		c_{\vec z,\vec h}=f_{s-1,\vec z}+\sum_{j=0}^{s-2}h_jf_{j,\vec z}.
		\]
		Consequently, every $c_{\vec z,\vec h}$ belongs to $V_s$. Conversely,
		\[
		f_{s-1,\vec z}=c_{\vec z,\vec0},\qquad
		f_{j,\vec z}=c_{\vec z,\vec e_j}+c_{\vec z,\vec0}\quad(0\le j<s-1).
		\]
		Thus the functions $c_{\vec z,\vec h}$ span $V_s$. For $s=1$ the second identity is absent and the first one suffices. Applying the linear evaluation map gives
		\[
		\operatorname{Col}(H^{\vec a})
		=\spn_{\F}\{\operatorname{ev}_{\mathcal M}(c_{\vec z,\vec h})\}
		=\operatorname{ev}_{\mathcal M}(V_s).
		\]
		The binary rank is the dimension of this column space, and evaluation is injective, so $\rk H^{\vec a}=\dim_{\F}V_s$.
	\end{proof}
	
	Put
	\[
	\mathsf{L}_{\vec{z}}(\vec x)=b+\sum_{i=1}^{m-1}2^{q_i}x_i\zeta_i,\qquad \vec x=(b,x_1,\dots,x_{m-1}),
	\]
	where $\zeta_i$ is the normalized coordinate \eqref{eq:zeta} of the column $\vec z$. Thus $\mathsf L_{\vec z}(\vec x)$ is the entry $W_{\vec x}(\vec{z})$ \eqref{eq:W} of the codeword $\vec u(\vec x)$ at the column $\vec z$. 
	
	\subsection{Boolean functions}\label{sec:anf}
	A message is a tuple $\vec x=(b,x_1,\ldots,x_{m-1})\in\mathcal M$. 
	A message function is any map $f\in\mathcal B=\{f:\mathcal M\to\F\}$. 
	Its input variables are the binary digits $b_r=\dig_r(b)$ for $0\le r<s$ and $a_{i,r}=\dig_r(x_i)$ for $0\le r<s-q_i$. 
	The column coordinates $\zeta_i$ are fixed parameters when a message function is formed. 
	For instance, $L_j=\dig_j\circ\mathsf L_{\vec z}$ is the message function giving digit $j$ at the fixed column $\vec z$.
	
	Writing the message bits collectively as $u_0,\ldots,u_{n-1}$, each $f$ has a unique algebraic normal form (ANF)
	\[
	f=\sum_{J\subseteq\{0,\ldots,n-1\}}c_J\prod_{\nu\in J}u_\nu,\qquad c_J\in\F.
	\]
	The squarefree monomials form a basis of $\mathcal B$. Both the ANF coefficient vector and the evaluation vector represent $f$ by linear isomorphisms. 
	In particular, elimination on ANF coefficients computes the same dimension as elimination on evaluation vectors.
	
	The \emph{support} $G(g)=\mathsf{supp}(g)$ of a squarefree monomial $g$ is the set of indices $i$ for which some $a_{i,r}$ occurs in $g$. 
	The bits of $b$ are allowed in $g$ and contribute no index to $G(g)$. 
	For example, $b_0a_{i,0}a_{i,1}$ has support $\{i\}$, whereas $b_0$ has empty support.
	
	For $U\subseteq\{1,\ldots,m-1\}$, let $e_U:\mathcal M\to\mathcal M$ replace $x_i$ by zero for each $i\in U$. 
	Its action on Boolean functions is the linear operator $E_U:\mathcal B\to\mathcal B$ given by $E_Uf=f\circ e_U$. 
	In particular, $E_i=E_{\{i\}}$, $E_\varnothing=\mathrm{id}$, and $E_U E_V=E_{U\cup V}$. For the digit function $L_j$ we have
	\[
	(E_i L_j)(\vec x)=(L_j\circ e_{\{i\}})(\vec x)
	=\dig_j\Bigl(b+\sum_{\ell\ne i}2^{q_\ell}x_\ell\zeta_\ell\Bigr).
	\]
	For a monomial $g$, $E_Ug=0$ if $U\cap G(g)\ne\varnothing$, and $E_Ug=g$ otherwise.
	
	\emph{The components of support.}  The projection onto the monomials of support exactly $S$ is
	\[
	\Pi_S=\sum_{U\subseteq S}E_{U^{\rm c}}=\sum_{T\supseteq S^{\rm c}}E_T,
	\qquad\text{equivalently}\qquad \Pi_S=\sum_{T\subseteq S}(-1)^{|S|-|T|}E_{T^{\rm c}},
	\]
	$S^{\rm c}$ being the complement in the set of all nonconstant factors (the signs are redundant in $\F$).  Indeed, for a monomial $g$ of support $G$, we have 
	\[
	\Pi_S g=\#\{U:G\subseteq U\subseteq S\}\,g=
	\begin{cases}
		0,&G\not\subseteq S,\\
		2^{|S|-|G|}g,&G\subseteq S, 
	\end{cases}
	\]
	because the number of $U$ with $G\subseteq U\subseteq S$ is $2^{|S|-|G|}$, which is even as soon as $G$ is a proper subset of $S$. 
	That is to say, 
	\[ \Pi_Sg=\begin{cases}g,&G=S,\\0,&G\ne S.\end{cases} \]
	 Consequently
	\[
	\Pi_S^2=\Pi_S,\qquad \Pi_S\Pi_{S'}=0\ \ (S\ne S'),\qquad \sum_S\Pi_S=\mathrm{id},
	\]
	so the components of support $S$ of a message function $f$ are the functions $\Pi_Sf$, and $f=\sum_S\Pi_Sf$.

	\subsection{Rank}
	
	For a column $\vec z\in\mathcal G$ let $\zeta_i$, $1\le i<m$, be its normalized coordinates \eqref{eq:zeta}. 
	For a subset $S \subseteq \{1,\dots,m-1\}$, a message $\vec{x} = (b,x_1,\dots,x_{m-1})$, and a digit $j<s$, define the following function
	\[
	D_{j,S,\zeta_S}(\vec x)=(\Pi_S L_j)(\vec x)
	=\sum_{U\subseteq S}\dig_j\Bigl(b+\sum_{i\in U}2^{q_i}x_i\zeta_i\Bigr).
	\]
	\begin{proposition}
		The Boolean function $D_{j,S,\zeta_S}$ belongs to $V_s\subseteq\mathcal B$. Consequently, its evaluation vector belongs to $\operatorname{Col}(H^{\vec a})$.
	\end{proposition}
	\begin{proof}
		For any set $V$ of non-constant factors, let $\vec z'$ be the canonical column obtained from $\vec z$ by setting its normalized coordinates $\zeta_i$ with $i\in V$ to zero. 
		This column exists because the canonical coordinate set is a product. 
		The generator $f_{j,\vec z}$ of $V_s$ satisfies $E_V f_{j,\vec z}=f_{j,\vec z'}$. 
		Thus each $E_V$ preserves $V_s$, as does $\Pi_S$. 
		In particular, $D_{j,S,\zeta_S}=\Pi_S f_{j,\vec z}\in V_s$. 
		The statement about evaluation follows from Lemma~\ref{lem:digitspan}. 
	\end{proof}
	
	Since $D_{j,S,\zeta_S}$ is a Boolean function of the message bits $b_0,\dots,b_{s-1}$ and $a_{i,r}$ with $0\le r<s-q_i$ and $i\in S$, it has a unique algebraic normal form
	\[
	D_{j,S,\zeta_S}(\vec{x}) = \sum_{J\subseteq \{0,\dots,s-1\}}\ \sum_{(B_i)_{i\in S},\ B_i\subseteq\{0,\dots,s-q_i-1\}} c_{J,(B_i)} \prod_{r\in J} b_r \prod_{i\in S}\prod_{r\in B_i} a_{i,r}.
	\]
	For a monomial $g$ of $D_{j,S,\zeta_S}$ the support $\mathsf{supp}(g)$ is defined as in \S\ref{sec:anf}: it is the set of factors $i\in S$ such that some bit $a_{i,r}$ occurs in $g$.
	
	\begin{lemma}\label{lem:directsum}
		Define the function subspace
		\[
		\mathcal W_s(S)=\spn_{\F}\left\{D_{j,S,\zeta_S}:0\le j<s,\quad
		\zeta_S\in\prod_{i\in S}\Z_{2^{s-q_i}}\right\}\subseteq V_s.
		\]
		Its evaluation image is precisely the span of the corresponding column vectors $\bigl(D_{j,S,\zeta_S}(\vec x)\bigr)_{\vec x\in\mathcal M}$, and the two spaces have the same dimension.
		Then 
		\begin{enumerate}[label=$(\arabic*)$]
			\item $\mathsf{supp}(g) = S$ for every monomial $g$ of $D_{j,S,\zeta_S}(\vec{x})$, 
			\item $V_s=\bigoplus_S\mathcal W_s(S)$ is an internal direct sum of function subspaces of $\mathcal B$,
			\item $\mathcal W_s(\varnothing)$ is spanned by the $s$ bits of $b$. 
		\end{enumerate}
	\end{lemma}
	\begin{proof}
		The canonical columns realize every allowed $\zeta_S$, so applying $\Pi_S$ to the digit generators gives $\mathcal W_s(S)=\Pi_S V_s$. 
		By Section~\ref{sec:anf}, $\Pi_S$ keeps exactly the monomials of support $S$, which proves (1). 
		The identity $\sum_S\Pi_S=\mathrm{id}$ gives the sum in (2), and the monomial basis makes it direct. 
		Finally, $\Pi_\varnothing L_j$ is obtained by setting all non-constant factors to zero, giving $\dig_j(b)=b_j$. This proves (3). 
	\end{proof}
	
	\begin{lemma}\label{lem:coeffrep}
		Let $T$ be an indeterminate and put
		\[
		B_s(T)=\prod_{r=0}^{s-1}(1+b_rT^{2^r}),\qquad
		A_{q,\zeta}(x;T)=\prod_{r=0}^{s-q-1}\Bigl(1+\alpha_r\bigl((1+T^{2^{q+r}})^{\zeta}-1\bigr)\Bigr),
		\]
		where $\alpha_r=\dig_r(x)$ are the bits of the coefficient $x\in\Z_{2^{s-q}}$ and $\zeta\in\Z_{2^{s-q}}$. 
		Choose standard nonnegative integer representatives $\tilde b$, $\tilde x_i$, and $\tilde\zeta_i$ of $b$, $x_i$, and $\zeta_i$, respectively, and put $\tilde L_{\vec z}(\vec x)=\tilde b+\sum_i2^{q_i}\tilde x_i\tilde\zeta_i$. This is a nonnegative integer lift of $\mathsf L_{\vec z}(\vec x)$ for every message $\vec x=(b,x_1,\dots,x_{m-1})$ and every canonical column $\vec z$.  
		
		Recall that $\mathcal B$ is the $\F$-algebra of functions on $\mathcal M$. Via the binary digits of messages, it is also the algebra of Boolean functions of the message bits, equivalently the polynomial algebra in these bits modulo $v^2=v$ for every bit $v$.
		Then the following identity holds in $\mathcal B[T]/(T^{2^s})$, or pointwise in $\F[T]/(T^{2^s})$ for each fixed message: 
		\[
		(1+T)^{\tilde L_{\vec z}(\vec x)}\equiv B_s(T)\prod_{i=1}^{m-1}A_{q_i,\zeta_i}(x_i;T)\pmod{T^{2^s}}.
		\]
		Moreover, for $0\le j<s$,
		\[
		D_{j,S,\zeta_S}=[T^{2^j}]\,B_s(T)\prod_{i\in S}\bigl(A_{q_i,\zeta_i}(x_i;T)-1\bigr),
		\]
		and $\mathcal W_s(S)=0$ whenever $\sum_{i\in S}2^{q_i}>2^{s-1}$.
	\end{lemma}
	\begin{proof}
		All identities are read in the ring of Boolean functions of the message bits with the formal variable $T$, truncated at $T^{2^s}$. We have the following steps. 
		\begin{enumerate}[label=(\alph*),leftmargin=*]
			\item For an integer $u\ge0$,
			\[
			(1+T)^u=\prod_{r:\ \dig_r(u)=1}(1+T^{2^r}),
			\]
			because $(1+T)^{2^r}=1+T^{2^r}$ in characteristic two. As the integers $\sum_{r\in J}2^r$ are distinct for distinct $J$, the coefficient of $T^{2^j}$ is $1$ exactly when $\dig_j(u)=1$, i.e.\ $[T^{2^j}](1+T)^u=\dig_j(u)$. 
			
			\item For a coefficient $x=\sum_{r<s-q}\alpha_r2^r$ and a column value $\zeta$,
			\[
			(1+T)^{2^qx\zeta}=\prod_{r=0}^{s-q-1}(1+T^{2^{q+r}})^{\alpha_r\zeta}
			=\prod_{r=0}^{s-q-1}\Bigl[1+\alpha_r\bigl((1+T^{2^{q+r}})^{\zeta}-1\bigr)\Bigr]=A_{q,\zeta}(x;T),
			\]
			where we use $\alpha_r\in\{0,1\}$, and 
			\[ (1+T^{2^{q+r}})^{\alpha_r\zeta} = 1+\alpha_r\bigl((1+T^{2^{q+r}})^{\zeta}-1\bigr). \]
			The product is finite because $\alpha_r=0$ for $r\ge s-q$.  For the all-one row, i.e.\ $q=0$, $\zeta=1$ and $x=b$, this gives $B_s(T)$.
			
			\item Use the integer representatives fixed in the statement, writing $\zeta_i$ for $\tilde\zeta_i$ in exponents, and put $\tilde L=\tilde b+\sum_i2^{q_i}\tilde x_i\zeta_i$.  
			Since $(1+T)^{2^s}=1+T^{2^s}$ has no term of degree below $2^s$, the coefficients of degree $<2^s$ of $(1+T)^{\tilde L}$ depend only on $\tilde L\bmod 2^s$. Hence
			\[
			(1+T)^{\tilde L_{\vec z}(\vec x)}=(1+T)^{\tilde b}\prod_{i=1}^{m-1}(1+T)^{2^{q_i}\tilde x_i\zeta_i}=B_s(T)\prod_{i=1}^{m-1}A_{q_i,\zeta_i}(x_i;T).
			\]
			
			\item   
			Put $P(\vec x;T)=(1+T)^{\tilde L(\vec x)}$. 
			Extend $E_V$ to $\mathcal B[T]/(T^{2^s})$ coefficient-wise. Then
			\[
			(E_VP)(\vec x;T)=P(e_V(\vec x);T) =(1+T)^{\tilde L(e_V(\vec x))}.
			\]
			In particular, $E_V$ commutes with extracting any coefficient $[T^d]$.
			
			\item For $U\subseteq S$, evaluating $P$ at $e_{U^{\rm c}}(\vec x)$ makes each factor with index outside $U$ equal to $1$. Therefore, in characteristic two,
			\[
			\begin{aligned}
				\sum_{U\subseteq S}E_{U^{\rm c}}P
				&=B_s(T)\sum_{U\subseteq S}\prod_{i\in U}A_{q_i,\zeta_i}(x_i;T)\\
				&=B_s(T)\prod_{i\in S}\bigl(A_{q_i,\zeta_i}(x_i;T)-1\bigr).
			\end{aligned}
			\]
			By (a) and (d), the coefficient of $T^{2^j}$ on the left is $\Pi_S L_j=D_{j,S,\zeta_S}$.
			
			\item The identity $A_{q,\zeta}(x;T)=(1+T^{2^q})^{\tilde x\tilde\zeta}$ shows that $A_{q,\zeta}(x;T)-1$ is divisible by $T^{2^q}$. 
			The product over $S$ is therefore divisible by $T^{\sum_{i\in S}2^{q_i}}$. 
			If this sum exceeds $2^{s-1}$, every $D_{j,S,\zeta_S}$ is zero, so $\mathcal W_s(S)=0$. 
		\end{enumerate}
		We complete the proof. 
	\end{proof}
	
	Recall that the row $w_{i+1}$ has level $q_i+1$, hence order $2^{s-q_i}$.  
	For a given $S \subseteq \{1,\dots,m-1\}$, define the multiplicity vector $\vec{k} = \vec{k}(S) = (k_0, k_1, \dots, k_{s-1})$ as follows  
	\[
	k_\ell := |\{ i \in S: q_i = \ell \}|, \qquad 0 \leq \ell < s. 
	\]
	
	\begin{lemma}\label{lem:equaldim}
		Let $S,S'$ be subsets of $\{1,\dots,m-1\}$ with $\vec{k} = \vec{k}(S) = \vec{k}(S')$.  
		Then $\dim \mathcal W_s(S)=\dim \mathcal W_s(S')$.
	\end{lemma}
	\begin{proof}
		Let $\sigma$ be a bijection of the set of nonconstant factors that maps each factor of $S$ to a factor of $S'$ of the same level. It exists because the two supports have the same number of factors at each level.  
		On the message, $\sigma$ relabels the coefficients within each level, which is the invertible linear substitution $T_\sigma$ of the space of Boolean functions given by $x_i\mapsto x_{\sigma(i)}$.  
		For a factor of level $q+1$ both its bits and its level are preserved, so $T_\sigma$ sends $D_{j,S,\zeta_S}$ to $D_{j,S',\zeta_{S'}}$: in that definition only the levels of the factors of $S$, and the products $2^{q_i}$, occur, and these are the same for the corresponding factor of $S'$.  
		Hence $T_\sigma$ maps $\mathcal W_s(S)$ onto $\mathcal W_s(S')$, and since $T_\sigma$ is invertible the two spaces have the same dimension.
	\end{proof}
	
	\begin{lemma}\label{lem:universal}
		Let $\vec k$ be feasible and let $S\subseteq\{1,\dots,m-1\}$ have $\vec k(S)=\vec k$.  Then $\dim\mathcal W_s(S)$ depends only on $s$ and $\vec k$, and neither on the type $\vec a$ nor on the length.
	\end{lemma}
	\begin{proof}
		Let $\mathcal M_{s,\vec k}=\Z_{2^s}\times\prod_{\ell=0}^{s-1}\Z_{2^{s-\ell}}^{k_\ell}$ be the local message box of the multiplicity vector and let $\pi_S$ be the projection of the ambient message box onto the coordinates $b$ and $x_i$, $i\in S$.  The functions $D_{j,S,\zeta_S}$ involve only these coordinates, so they are pullbacks $\pi_S^*\tilde D_{j,S,\zeta_S}$ of functions $\tilde D$ of the local box, and $D\mapsto\pi_S^*D$ is injective because $\pi_S$ is surjective.  An injective linear map preserves the dimension of every span, so $\dim\mathcal W_s(S)$ equals the dimension of the span of the $\tilde D_{j,S,\zeta_S}$, which depends only on $s$ and $\vec k$.  This also gives a meaning to $\rho_s(\vec k)$ for multiplicity vectors that do not occur in a given type.
	\end{proof}
	
	Let $\rho_s(\vec k) = \dim \mathcal W_s(S)$ for $S$ with $\vec{k} = \vec{k}(S)$.  
	It is clear that $\rho_s(\vec k)$ is well-defined and depends only on $s$ and $\vec k$, by Lemma \ref{lem:equaldim} and Lemma \ref{lem:universal}. 
	Using Lemma \ref{lem:directsum}, $\rho_s(\vec 0)=\dim \mathcal{W}(\varnothing) = s$.
	
	\begin{theorem}\label{thm:rankformula}
		For the type $\vec{a} = (a_1, \dots, a_s)$, define $\mu_0 = a_1-1$ and $\mu_\ell = a_{\ell+1}$ for $1 \leq \ell < s$.  Then
		\[
		\rk H^{\vec a}=\sum_{\substack{\vec k\in\Z_{\ge 0}^s\\ \sum_{\ell=0}^{s-1} 2^\ell k_\ell \le 2^{s-1}}}
		\rho_s(\vec k)\prod_{\ell=0}^{s-1}\binom{\mu_\ell}{k_\ell}, 
		\]
		where $\binom{n}{0}=1$ for every $n\ge0$, and $\binom{n}{k}=0$ for integers $k>n\ge 0$. 
	\end{theorem}
	\begin{proof}
		By Lemma~\ref{lem:digitspan} the rank is $\dim V_s$, and by Lemma~\ref{lem:directsum},
		\[
		\rk H^{\vec a} = \dim V_s=\sum_{S}\dim \mathcal W_s(S),
		\]
		the sum running over all subset $S$ of the set $\{1,\dots,m-1\}$.  
		For given $S$, let $\vec{k} = \vec{k}(S) = (k_0, k_1, \dots, k_{s-1})$ and $w(\vec{k}) = \sum_{\ell=0}^{s-1} 2^\ell k_\ell$. 
		Grouping the sum over the supports according to their multiplicity vector $\vec k(S)$, and using Lemma~\ref{lem:coeffrep} and Lemma~\ref{lem:equaldim}, we get
		\begin{align*}
			\rk H^{\vec a} 
			& = \sum_{\vec{k}} \sum_{S: \vec{k}(S) = \vec{k}} \dim\mathcal W_s(S) \\
			& = \sum_{\vec{k}:w(\vec{k}) \leq 2^{s-1}} \rho_s(\vec{k})\cdot |\{ S:\vec{k}(S) = \vec{k} \}| 
		\end{align*}
		because $\dim\mathcal W_s(S)=\rho_s(\vec k(S))$ for every $S$ and it vanishes when $w(\vec k(S))>2^{s-1}$, so each support contributes exactly one summand $\rho_s(\vec k(S))$.
		By the definition of $\vec k(S)$ there are
		\[
		|\{ S:\vec{k}(S) = \vec{k} \}| = \prod_{\ell=0}^{s-1}\binom{\mu_\ell}{k_\ell}
		\]
		supports of multiplicity $\vec k$: one chooses $k_\ell$ of the $\mu_\ell$ factors of level $\ell+1$ for each $\ell$, and the choices at different levels are independent.
	\end{proof}
	
	\begin{remark}
		For fixed $s$, the coefficient table can be precomputed once. After that, evaluating the rank for a given type takes at most $N_s$ summands of the binomial formula. The number of feasible multiplicity vectors is
		\[
		N_s=\bigl|\{\vec k:\textstyle\sum_\ell2^\ell k_\ell\le2^{s-1}\}\bigr|=4,\ 10,\ 36,\ 202,\dots,
		\]
		as $s=2,3,4,5,\dots$. 
		The cost of precomputation also depends on the sizes of the polynomial spans in Appendix~\ref{app:algo}.
	\end{remark}
	
	\subsection{Relation to the kernel-block rank profile}\label{subsec:profile}
	
	Let $C$ be a binary code containing $\vec 0$.
	Two coordinate positions $p$ and $q$ belong to the same \emph{kernel block} if $u_p=u_q$ for every $\vec u\in\Ker(C)$.
	For a kernel block $B$, write $C|_B$ for the punctured code obtained by retaining precisely the coordinates in $B$.
	The \emph{local rank} at $B$ is $\rk(C|_B)$, and the \emph{kernel-block rank profile} is the multiset of these local ranks over all kernel blocks. 
	See \cite{Profile} for more details. 
	
	Let $C=H^{a_1,\dots,a_s}$ be a nonlinear $\Z_{2^s}$-linear Hadamard code with $a_1\ge2$, of length $2^t$ and kernel dimension $\kappa=\dim\Ker(C)$.   
	By \cite[Corollary 5.3]{Profile} and Theorem \ref{thm:rankformula}, its constant local rank is 
	\[
	\rk H^{a_1,\dots,a_{s-1}}=\sum_{\substack{\vec k\in\Z_{\ge0}^{s-1}\\ \sum_{\ell=0}^{s-2}2^\ell k_\ell\le2^{s-2}}}
	\rho_{s-1}(\vec k)\prod_{\ell=0}^{s-2}\binom{\mu_\ell}{k_\ell},
	\]
	which does not depend on $a_s$, and each block has $2^{t-\kappa+1}$ coordinates, and $\mu_0=a_1-1$ and $\mu_\ell=a_{\ell+1}$ for $1\le\ell\le s-2$.  
	
	\section{The values of \texorpdfstring{$\rho_s(\vec k)$}{rho}}\label{sec:rho}
	
	In this section, we will discuss the values of $\rho_s(\vec{k})$. 
	
	First, we give the following relation. 
	\begin{proposition}\label{prop:levelshift}
		Let $s\ge2$, and let $\vec{k}=(0,k_1,\dots,k_{s-1})$, $\vec{k}'=(k_1,\dots,k_{s-1})$. 
		If $\vec{k}'\ne\vec{0}$, then
		\[ \rho_s(\vec{k})=\rho_{s-1}(\vec{k}'). \]
	\end{proposition}
	\begin{proof}
		Work on the local message domain of Lemma~\ref{lem:universal}. Every $q_i\ge1$, so write $b=b_0+2\beta$ and set
		\[
		\mathsf L'=\beta+\sum_{i\in S}2^{q_i-1}x_i\zeta_i.
		\]
		Then $\mathsf L=b_0+2\mathsf L'$, whence $\dig_0(\mathsf L)=b_0$ and $\dig_j(\mathsf L)=\dig_{j-1}(\mathsf L')$ for $1\le j<s$. Since $S\ne\varnothing$, the exact-support projection kills $b_0$. Let $S'$ be the same set of factors with offsets $q_i-1$, and consider the surjective map
		\[
		p:\Z_{2^s}\times\prod_{i\in S}\Z_{2^{s-q_i}}
		\longrightarrow\Z_{2^{s-1}}\times\prod_{i\in S}\Z_{2^{s-q_i}},
		\qquad (b,(x_i))\longmapsto(\lfloor b/2\rfloor,(x_i)).
		\]
		The coefficient and parameter domains are unchanged, since $(s-1)-(q_i-1)=s-q_i$. The digit identities give
		\[
		D^{(s)}_{j,S,\zeta_S}=
		\begin{cases}0,&j=0,\\p^*D^{(s-1)}_{j-1,S',\zeta_S},&1\le j<s,\end{cases}
		\]
		where $p^*f=f\circ p$ and superscripts indicate the alphabet exponent. Thus the local spaces satisfy $\widetilde{\mathcal W}_s(S)=p^*\widetilde{\mathcal W}_{s-1}(S')$. Pullback by a surjective map is injective, and passing to local domains preserves dimension by Lemma~\ref{lem:universal}. Therefore $\rho_s(\vec k)=\rho_{s-1}(\vec k')$.
	\end{proof}
	
	\begin{remark}
		The nonempty-support condition cannot be dropped, since $\rho_s(\vec0)=s$ whereas $\rho_{s-1}(\vec 0)=s-1$.
	\end{remark}
	
	Next, we present the values of $\rho_s(\vec{k})$, whose proof can be seen in Appendix \ref{app:rho}. 
	\begin{theorem} \label{thm:rho}
		Let $K = 2^{s-1}$. 
		Then 
		\[
		\rho_s(\vec{k}) = 
		\begin{cases}
			s, & \vec{k} = \vec{0}, \\
			1, & w(\vec{k}) = K, \\
			2, & w(\vec{k}) = K-1, \ s \geq 2, \\
			k_0 + 2, & w(\vec{k}) = K-2, \ s \geq 3, \\
			3 k_0 + 2 -\varepsilon, & w(\vec{k}) = K-3, \ s\geq 4, 
		\end{cases}
		\]
		where $\varepsilon = 1$ for $(k_0, k_1) = (1,0)$, and $0$ for otherwise. 
		Moreover, if $\vec{k} = (1,0,\dots,0)$, then 
		\[ \rho_s(\vec{k}) =  \left\lfloor \frac{2^{s+1}}{3} \right\rfloor. \]
	\end{theorem}
	
	\begin{example}
		For $s=2$, we have $K=2$. 
		The feasible vectors in Theorem~\ref{thm:rankformula} are $(0,0)$, $(1,0)$, $(2,0)$, and $(0,1)$. 
		Using Theorem \ref{thm:rho}, the related $\rho_2$ are $2,2,1,1$, respectively.  
		Hence, 
		\begin{equation}\label{eq:r2}
			\rk H^{a_1,a_2}=2+2(a_1-1)+\binom{a_1-1}{2}+a_2=2a_1+a_2+\binom{a_1-1}{2},
		\end{equation}
		which can also be seen in \cite{Z8}. 
	\end{example}
	
	\begin{example}
		For $s=3$, Table~\ref{tab:rho3} lists the values given by Theorem~\ref{thm:rho}. 
		Indeed, $K=4$. 
		Thus, by Theorem \ref{thm:rho}, the values of $\rho_s(\vec{k})$ can be determined for all $\vec{k}$ since $0 \leq w(\vec{k}) \leq 4=K$. 
		Carrying out the sum gives the closed formula of \cite{Z8},
		\[
		\rk H^{a,b,c}=\frac{a^4-2a^3+35a^2+14a}{24}+1+\frac b2\bigl(a^2+a+b+1\bigr)+c.
		\]
		\begin{table}[ht]
			\centering
			\caption{The coefficients $\rho_3(\vec k)$}\label{tab:rho3}
			\begin{tabular}{ccc}
				\toprule
				$(k_1,k_2)$ & $k_0$ & $\rho_3(\vec k)$ \\
				\midrule
				$(0,0)$ & $0$--$4$ & $3, 5, 4, 2, 1$ \\
				$(1,0)$ & $0$--$2$ & $2, 2, 1$ \\
				$(0,1)$ & $0$ & $1$ \\
				$(2,0)$ & $0$ & $1$ \\
				\bottomrule
			\end{tabular}
		\end{table}
	\end{example}
	
	\begin{example}
		For $s=4$ and $s=5$, the related coefficients $\rho_4$ and $\rho_5$ are listed in Table~\ref{tab:rho4} (36 feasible supports, of weights at most $2^3=8$), and Table~\ref{tab:rho5} (202 feasible supports, of weights at most $2^{s-1}=16$), respectively. 
		The coefficients $\rho_s(\vec{k})$ are computed by Algorithm \ref{alg:algo}. 
		\begin{table}[htbp]
			\centering
			\caption{The coefficients $\rho_4(\vec k)$}\label{tab:rho4}
			\begin{tabular}{@{}ccl@{}}
				\toprule
				$(k_1,k_2,k_3)$ & $k_0$ & $\rho_4(\vec k)$ \\
				\midrule
				$(0,0,0)$ & $0$--$8$ & $4, 10, 20, 27, 24, 17, 8, 2, 1$ \\
				$(1,0,0)$ & $0$--$6$ & $5, 9, 12, 11, 6, 2, 1$ \\
				$(0,1,0)$ & $0$--$4$ & $2, 4, 4, 2, 1$ \\
				$(2,0,0)$ & $0$--$4$ & $4, 5, 4, 2, 1$ \\
				$(0,0,1)$ & $0$ & $1$ \\
				$(1,1,0)$ & $0$--$2$ & $2, 2, 1$ \\
				$(3,0,0)$ & $0$--$2$ & $2, 2, 1$ \\
				$(0,2,0)$ & $0$ & $1$ \\
				$(2,1,0)$ & $0$ & $1$ \\
				$(4,0,0)$ & $0$ & $1$ \\
				\bottomrule
			\end{tabular}
		\end{table}
		\begin{table}[htbp]
			\centering
			\caption{The coefficients $\rho_5(\vec k)$} \label{tab:rho5}
			\begin{tabular}{@{}cc>{\raggedright\arraybackslash}p{108mm}@{}}
				\toprule
				$(k_1,\dots,k_{4})$ & $k_0$ & $\rho_5(\vec{k})$\\
				\midrule
				$(0,0,0,0)$ & $0$--$16$ & $5, 21, 79, 210, 420, 689, 949, 1081, 1012, 774, 488, 267, 116, 41, 16, 2, 1$ \\
				$(1,0,0,0)$ & $0$--$14$ & $10, 41, 111, 219, 340, 443, 475, 420, 306, 185, 88, 35, 14, 2, 1$ \\
				$(0,1,0,0)$ & $0$--$12$ & $5, 18, 46, 88, 138, 168, 158, 115, 62, 29, 12, 2, 1$ \\
				$(2,0,0,0)$ & $0$--$12$ & $20, 56, 110, 165, 198, 206, 176, 119, 64, 29, 12, 2, 1$ \\
				$(0,0,1,0)$ & $0$--$8$ & $2, 8, 19, 27, 24, 17, 8, 2, 1$ \\
				$(1,1,0,0)$ & $0$--$10$ & $9, 24, 45, 66, 76, 65, 42, 23, 10, 2, 1$ \\
				$(3,0,0,0)$ & $0$--$10$ & $27, 53, 77, 92, 90, 69, 44, 23, 10, 2, 1$ \\
				$(0,0,0,1)$ & $0$ & $1$ \\
				$(0,2,0,0)$ & $0$--$8$ & $4, 10, 20, 27, 24, 17, 8, 2, 1$ \\
				$(1,0,1,0)$ & $0$--$6$ & $4, 9, 12, 11, 6, 2, 1$ \\
				$(2,1,0,0)$ & $0$--$8$ & $12, 23, 30, 31, 26, 17, 8, 2, 1$ \\
				$(4,0,0,0)$ & $0$--$8$ & $24, 37, 40, 35, 28, 17, 8, 2, 1$ \\
				$(0,1,1,0)$ & $0$--$4$ & $2, 4, 4, 2, 1$ \\
				$(1,2,0,0)$ & $0$--$6$ & $5, 9, 12, 11, 6, 2, 1$ \\
				$(2,0,1,0)$ & $0$--$4$ & $4, 5, 4, 2, 1$ \\
				$(3,1,0,0)$ & $0$--$6$ & $11, 13, 14, 11, 6, 2, 1$ \\
				$(5,0,0,0)$ & $0$--$6$ & $17, 17, 16, 11, 6, 2, 1$ \\
				$(0,0,2,0)$ & $0$ & $1$ \\
				$(0,3,0,0)$ & $0$--$4$ & $2, 4, 4, 2, 1$ \\
				$(1,1,1,0)$ & $0$--$2$ & $2, 2, 1$ \\
				$(2,2,0,0)$ & $0$--$4$ & $4, 5, 4, 2, 1$ \\
				$(3,0,1,0)$ & $0$--$2$ & $2, 2, 1$ \\
				$(4,1,0,0)$ & $0$--$4$ & $6, 5, 4, 2, 1$ \\
				$(6,0,0,0)$ & $0$--$4$ & $8, 5, 4, 2, 1$ \\
				$(0,2,1,0)$ & $0$--$0$ & $1$ \\
				$(1,3,0,0)$ & $0$--$2$ & $2, 2, 1$ \\
				$(2,1,1,0)$ & $0$ & $1$ \\
				$(3,2,0,0)$ & $0$--$2$ & $2, 2, 1$ \\
				$(4,0,1,0)$ & $0$ & $1$ \\
				$(5,1,0,0)$ & $0$--$2$ & $2, 2, 1$ \\
				$(7,0,0,0)$ & $0$--$2$ & $2, 2, 1$ \\
				$(0,4,0,0)$ & $0$ & $1$ \\
				$(2,3,0,0)$ & $0$ & $1$ \\
				$(4,2,0,0)$ & $0$ & $1$ \\
				$(6,1,0,0)$ & $0$ & $1$ \\
				$(8,0,0,0)$ & $0$ & $1$ \\
				\bottomrule
			\end{tabular}
		\end{table}
	\end{example}

	\section{Applications of the rank invariant} \label{sec:app}
	
	The total rank alone need not classify the $\mathbb Z_{2^s}$-linear Hadamard codes. For each fixed $s$, however, the ranks along the nonlinear descent give a complete invariant for the nonlinear reduced family.  
	The key is the following consequence of Theorem~\ref{thm:rankformula}.
	
	\begin{lemma}\label{lem:lastparam}
		For $s\ge2$ and fixed $a_1,\dots,a_{s-1}$,
		\[
		\rk H^{a_1,\dots,a_{s-1},a_s} = a_s+\rk H^{a_1,\dots,a_{s-1},0}.
		\]
	\end{lemma}
	\begin{proof}
		Let $\vec{k} = (k_0, \dots, k_{s-1})$. 
		If $k_{s-1} = 0$, then by Theorem~\ref{thm:rankformula},  we have 
		\[ \binom{\mu_{s-1}}{k_{s-1}} = \binom{a_s}{0} = 1 = \binom{0}{0} \]
		which is the same as the case $\mu_{s-1} = 0$. 
		Hence, the corresponding terms do not involve $a_s$, and they are exactly the terms of $\rk H^{a_1,\dots,a_{s-1},0}$.  
		
		If $k_{s-1} = 1$, then $w(\vec{k}) \leq 2^{s-1}$ if and only if $\vec{k} = (0,\dots,0,1)$, i.e., $k_0=\cdots=k_{s-2}=0$. 
		Moreover, we have $\rho_s(0,\dots,0,1) = 1$ by Proposition~\ref{prop:top}.  
		Then 
		\begin{align*}
			\rk H^{a_1,\dots,a_{s-1},a_s}
			& =\sum_{\vec{k}:k_{s-1}=1} \rho_s(\vec k) \prod_{\ell=0}^{s-1}\binom{\mu_\ell}{k_\ell} 
			+ \sum_{\vec{k}:k_{s-1}=0} \rho_s(\vec k) \prod_{\ell=0}^{s-1}\binom{\mu_\ell}{k_\ell} \\
			& =\rk H^{a_1,\dots,a_{s-1},0} + \rho_s(0,\dots,0,1) \binom{a_{s}}{1} \\
			& = a_s + \rk H^{a_1,\dots,a_{s-1},0}, 
		\end{align*}
		which completes the proof. 
	\end{proof}
	
	Let $C=H^{a_1,\ldots,a_s}$ be nonlinear and reduced, so $a_1\ge2$.
	For a nonlinear code $D=H^{a_1,\ldots,a_j}$ occurring in the descent, define its derived code $\partial D$ to be $D|_B$ for any kernel block $B$ of $D$. 
	By \cite[Theorem~5.2]{Profile}, puncturing on the distinguished block $B_0$ gives a code permutation equivalent
	to $H^{a_1,\ldots,a_{j-1}}$. By \cite[Theorem~4.5]{Profile}, the punctured codes on any two kernel blocks are permutation equivalent. 
	Thus $\partial D$ is well defined up to coordinate permutation, and it has type $(a_1,\ldots,a_{j-1})$ in its canonical representation.
	
	Set $C_0=C$ and $C_{i+1}=\partial C_i$ as long as $C_i$ is nonlinear.
	Let $\mathfrak m$ be the first index for which $C_{\mathfrak m}$ is linear, and put $L=C_{\mathfrak m}$. 
	Define
	\[
	I_s(C)=\bigl(\rk C_0,\rk C_1,\ldots,\rk C_{\mathfrak m}\bigr).
	\]
	
	To see that this sequence is invariant under binary code equivalence, let $D'=\pi(D)+\vec v$, where $\pi$ is a coordinate permutation and both $D$ and $D'$ contain $\vec 0$.
	Since $\Ker(D')=\pi(\Ker(D))$, the permutation $\pi$ maps each kernel block $B$ of $D$ to a kernel block $\pi(B)$ of $D'$.
	Writing $\pi_B$ for the induced coordinate bijection, we have
	\[
	D'|_{\pi(B)} =\pi_B(D|_B)+\vec v|_{\pi(B)}.
	\]
	Both punctured codes contain $\vec 0$, so the translation vector belongs to $\pi_B(D|_B)$. Translation by this vector preserves the binary span and the property of being linear. 
	Applying this argument at each step shows that both the stopping index $\mathfrak m$ and the rank sequence $I_s(C)$ are invariant under equivalence.
	
	If $a_1\ge3$, then all $H^{a_1,\dots,a_j}$ with $j\ge2$ are nonlinear and $H^{a_1}$ is linear, so ${\mathfrak{m}}=s-1$ and $L=H^{a_1}$.  
	If $a_1=2$ and $s\ge3$, then $H^{2,a_2}$ is linear while $H^{2,a_2,\dots,a_j}$ is nonlinear for $j\ge3$, so ${\mathfrak{m}}=s-2$ and $L=H^{2,a_2}$.  
	Thus, for a fixed $s$, the value of ${\mathfrak{m}}$ distinguishes the two cases.
	
	\begin{proposition}
		For each fixed $s$, the sequence $I_s(C)$ uniquely determines the type $(a_1,\dots,a_s)$ of every nonlinear reduced code $C=H^{a_1,\dots,a_s}$.
	\end{proposition}
	\begin{proof}
		Write $r_k=\rk C_k$ for $0\le k\le\mathfrak m$, so $r_{\mathfrak{m}}=\rk C_{\mathfrak{m}} = \rk L$.
		
		If ${\mathfrak{m}}=s-1$, then $L=H^{a_1}$ and $a_1=r_{\mathfrak{m}}$.  
		If ${\mathfrak{m}}=s-2$, then $L=H^{2,a_2}$, and the $\Z_4$ rank formula \eqref{eq:r2} gives $r_{\mathfrak{m}}=4+a_2$. 
		Hence $a_1=2$ and $a_2=r_{\mathfrak{m}}-4$.
		
		If $a_1\ge3$, proceed with $j=2,\dots,s$. If $a_1=2$, the terminal rank has already determined $a_2$, so proceed with $j=3,\dots,s$. At each of these steps, Lemma~\ref{lem:lastparam} applied to $H^{a_1,\dots,a_j}$ gives
		\[
		a_j=r_{s-j}-\rk H^{a_1,\dots,a_{j-1},0},
		\]
		where $C_{s-j}$ is equivalent to $H^{a_1,\dots,a_j}$ and the baseline $H^{a_1,\dots,a_{j-1},0}$ is taken over $\Z_{2^j}$, so it is determined by the parameters already recovered.  
		Since $a_1,\dots,a_{j-1}$ are already known, this determines $a_j$.  Induction gives all $a_1,\dots,a_s$.
	\end{proof}
	
	The new point of this result is that the rank sequence of the nonlinear descent recovers the type intrinsically. 
	For instance, the two codes $H^{2,8,0,2}$ and $H^{5,0,7,0}$, which have the same total rank $674$ and the same kernel dimension $13$, have derived rank vectors
	\[
	(674,68,12)\qquad\text{and}\qquad(674,63,16,5),
	\]
	The first sequence stops at the linear code $H^{2,8}$. 
	In the second sequence, $\rk H^{5,0}=2\cdot5+\binom{4}{2}=16$ by \eqref{eq:r2}. 
	
	To compare codes with the same stopping index, consider $H^{2,1,0,4,0}$ and $H^{2,0,3,1,1}$. Both have length $2^{21}$ by \eqref{eq:length} and kernel dimension $8$ by the kernel dimension formula for nonlinear reduced codes in \cite{CFV2019}. Theorem~\ref{thm:rankformula} and the coefficient tables give
	\[
	I_5(H^{2,1,0,4,0})=(175,32,12,5),\qquad
	I_5(H^{2,0,3,1,1})=(175,36,11,4).
	\]
	Both sequences stop at depth $3$. The ranks at depth $1$ distinguish these codes even though the two derived codes at that depth have the same length $2^{14}$.
	
	The proposition gives a complete invariant consisting of ranks for the fixed-$s$ nonlinear reduced family. 
	The equivalence criterion for reduced types in \cite[Theorem~7.2]{Bhunia2026} provides the classification used here.
	
	\begin{theorem}[\cite{Bhunia2026}]\label{thm:bhunia}
		Let $H^{\vec{a}}$ and $H^{\vec{a}'}$ be two reduced codes of the same length. 
		Then $H^{\vec{a}}$ and $H^{\vec{a}'}$ are equivalent if and only if $\vec{a} = \vec{a}'$, where $\vec{a} = (a_1,\dots,a_s)$ and $\vec{a}' = (a'_1,\dots,a'_{s'})$. 
	\end{theorem}
	
	\begin{remark}
		Codes with $a_1=1$ are not reduced.  By \cite[Proposition 3.5]{Bhunia2026}, every nonlinear such code is permutation equivalent to a reduced code of the same length, obtained by applying the reduction $R$ of \cite[Lemma 3.3]{Bhunia2026}; the latter has $s-k$ parameters, first parameter at least $2$, and last parameter $a_s+k$.  The invariant above is applied to this reduced representative, and the linear code $H^{t+1}$ is treated separately.  
		Hence the nonlinear cases with $a_1=1$ can be studied after reduction. The completeness statement proved above is applied within each fixed alphabet exponent of the reduced representative; it does not assert completeness across different exponents without additional information.
	\end{remark}
	
	\section{Conclusion}\label{sec:conclusion}
	We studied the binary rank of the canonical $\Z_{2^s}$-linear Hadamard codes $H^{\vec a}$.  The main result is a representation theorem: for every type, the rank is a binomial polynomial in the parameters $\mu_0=a_1-1$ and $\mu_\ell=a_{\ell+1}$,
	\[
	\rk H^{\vec a}=\sum_{\vec k}\rho_s(\vec k)\prod_{\ell=0}^{s-1}\binom{\mu_\ell}{k_\ell},
	\]
	whose coefficients $\rho_s(\vec k)$ depend only on the alphabet exponent $s$ and on the multiplicity vector $\vec k$, and neither on the type nor on the length (Theorem~\ref{thm:rankformula} and Lemma~\ref{lem:universal}).  
	The coefficients $\rho_s(\vec{k})$ are produced by a finite exact algorithm (Appendix~\ref{app:algo}), based on the coefficient representation of Lemma~\ref{lem:coeffrep}: each local space is generated by the coefficients of $T^{2^j}$ in the product $B_s(T)\prod_{i\in S}(A_{q_i,\zeta_i}(x_i;T)-1)$, and these polynomials are computed by truncated multiplication over $\F$ followed by elimination on monomials.
	
	For $s=2$ and $s=3$ the sum collapses to compact closed forms, namely $2a_1+a_2+\binom{a_1-1}{2}$ and the quartic formula for $\Z_8$, both from \cite{Z8}.  
	For $s=4$ and $s=5$ one obtains explicit binomial polynomials in the type parameters whose coefficients are tabulated, with $36$ and $202$ feasible supports respectively (Tables~\ref{tab:rho4} and~\ref{tab:rho5}). 
	For arbitrary $s$ we do not have a simple closed formula uniformly in $s$.  
	What is proved is the representation theorem above together with a finite exact algorithm that computes $\rho_s(\vec k)$ for each fixed $s$.
	
	Two questions remain open: an evaluation of $\rho_s(\vec{k})$ in closed form for $\vec{k}$, and the ranks of the related families of mixed type. 
	
\appendix
	\section{The algorithm to compute \texorpdfstring{$\rho_s(\vec{k})$}{}}\label{app:algo}
	We state the algorithm in pseudocode and prove that it computes $\rho_s(\vec k)$.
	The pseudocode uses the same factor-space compression as the reference implementation. Each intermediate basis is a basis of whole truncated polynomials, with coordinates indexed jointly by the $T$-degree and the message monomial; coefficients of different $T$-degrees are not separated during this compression. 
	\begin{algorithm}[htbp]
		\caption{The coefficients of the rank formula}\label{alg:algo}
		\scriptsize
		\begin{algorithmic}[1]
			\Require $s\ge1$ and $\vec k\in\Z_{\ge0}^s$.
			\Ensure $\rho_s(\vec k)$, with value $0$ if $w(\vec k)>2^{s-1}$.
			\State $K\gets2^{s-1}$; $w\gets\sum_\ell2^\ell k_\ell$.
			\If{$w>K$}\State \Return $0$\EndIf
			\If{$\vec k=\vec0$}\State \Return $s$\EndIf
			\State Choose $S$ with $k_\ell$ factors of offset $\ell$; give distinct bit variables to different factors and to $b$.
			\State $\mathcal P\gets\{1\}$.
			\For{each factor $i\in S$}
			\State $d_i\gets K-w+2^{q_i}$; $\mathcal F_i\gets\varnothing$.
			\For{each $\zeta\in\{1,\dots,2^{s-q_i}-1\}$}
			\State $Q\gets1$.
			\For{$r=0,\dots,s-q_i-1$}
			\State $Q\gets Q\left[1+a_{i,r}\bigl((1+T^{2^{q_i+r}})^\zeta-1\bigr)\right]\bmod T^{d_i+1}$.
			\EndFor
			\State $F\gets Q-1$; add $F$ to $\mathcal F_i$.
			\EndFor
			\State $\mathcal F_i\gets$ a basis of its span, using the common coordinates $T^d m$ ($m$ a squarefree message monomial).
			\State $\mathcal Q\gets\{fg\bmod T^{K+1}:f\in\mathcal P,\ g\in\mathcal F_i\}$.
			\State $\mathcal P\gets$ a basis of $\spn_{\F}\mathcal Q$, using the common coordinates $T^d m$.
			\EndFor
			\State $\mathcal D\gets\varnothing$.
			\For{each $f\in\mathcal P$ and each $j=0,\dots,s-1$}
			\State Add $[T^{2^j}](B_s(T)f\bmod T^{K+1})$ to $\mathcal D$.
			\EndFor
			\State \Return $\dim_{\F}\spn\mathcal D$, by elimination on the squarefree message monomials.
		\end{algorithmic}
	\end{algorithm}

	\begin{theorem}\label{thm:algo}
		For every $s\ge1$ and every $\vec k$ with $w(\vec k)\le2^{s-1}$ the output of Algorithm~\ref{alg:algo} is $\rho_s(\vec k)=\dim\mathcal W_s(S)$.
	\end{theorem}
	\begin{proof}
		The cases $w>K$ and $\vec k=\vec 0$ follow from Lemmas~\ref{lem:coeffrep} and~\ref{lem:directsum}, respectively. Suppose $0<w\le K$. Let $F_{i,\zeta}=A_{q_i,\zeta}(x_i;T)-1$. Multiplying the displayed local factors first produces $A_{q_i,\zeta}$; the explicit subtraction of $1$ produces $F_{i,\zeta}$. The omitted parameter $\zeta=0$ gives the zero polynomial and does not affect its span.
		
		Every $F_{h,\zeta}$ has no term of degree below $2^{q_h}$. Consequently, a term of $F_{i,\zeta}$ of degree greater than $d_i=K-w+2^{q_i}$ cannot contribute to a full product of degree at most $K$, because the remaining factors have total degree at least $w-2^{q_i}$. Thus the local truncations are harmless, as are all later truncations modulo $T^{K+1}$.
		
		By induction over the processed factors, $\mathcal P$ is a basis of the span of all their products. Indeed, replacing each factor family by a basis of its polynomial span preserves the span of all pairwise products by bilinearity. This argument applies to whole polynomials in the common coordinates $T^d m$; it would not justify discarding the relations between coefficients of different degrees. After every factor is processed, $\mathcal P$ therefore spans all $\prod_{i\in S}F_{i,\zeta_i}$ modulo $T^{K+1}$.
		
		For each $j<s$, the map $f\mapsto[T^{2^j}](B_sf)$ is linear. Applying these maps to a basis yields the same combined coefficient span as applying them to every original product. Lemma~\ref{lem:coeffrep} identifies the latter coefficients with all the functions $D_{j,S,\zeta_S}$. The final elimination on the unique squarefree monomial coordinates therefore returns $\dim\mathcal W_s(S)=\rho_s(\vec k)$.
	\end{proof}
	
	\subsection*{Reproducing the coefficient tables}
	
	The reference implementation of Algorithm~\ref{alg:algo} is provided in \texttt{rho\_hadamard.py}, and the coefficient tables are stored in \texttt{rho\_tables.json}.  
	Run \texttt{python check\_tables.py} to verify the $50$ coefficients for $s=2,3,4$, or \texttt{python check\_tables.py --s5} to verify all $252$ coefficients for $s=2,3,4,5$. 
	The script exits with a nonzero status if an entry is missing, unexpected, or different from the recomputed value.

	\section{Proofs of the special values of \texorpdfstring{$\rho_s(\vec k)$}{Proofs of the special values of rho}}\label{app:rho}
	
	We use the message functions, offsets $q_i$, normalized column parameters $\zeta_i$, and local spaces $\mathcal W_s(S)$ defined in Section~\ref{sec:rank}. 
	By Lemma~\ref{lem:universal}, unused message coordinates may be omitted without changing their dimensions. Throughout this section, put
	\[
	K=2^{s-1},\qquad w(\vec k)=\sum_{i\in S}2^{q_i}=\sum_{\ell=0}^{s-1}2^\ell k_\ell.
	\]
	Unless a result states a stronger restriction, $s\ge1$. All coefficients are Boolean functions over $\F$, with $v^2=v$ for each message bit $v$. 
	For standard nonnegative representatives, set
	\[
	F_i(T)=(1+T)^{2^{q_i}x_i\zeta_i}-1.
	\]
	The coefficient identity of Lemma~\ref{lem:coeffrep} can equivalently be written as
	\begin{equation}\label{eq:coefficient}
		D_{j,S,\zeta_S}=[T^{2^j}]B_s(T)\prod_{i\in S}F_i(T),
		\qquad B_s(T)=(1+T)^b.
	\end{equation}
	Here exponents use the chosen integer representatives and equality is understood to the degree needed. Each $F_i$ has no term below $2^{q_i}$, so the product has no term below $w(\vec k)$. 
	The empty-support value $\rho_s(\vec0)=s$ follows from Lemma~\ref{lem:directsum}.
	
	\subsection{Coefficients}
	For a support $S$, let
	\[
	I=\{i\in S:q_i=0\},\qquad J=\{i\in S:q_i=1\},\qquad
	R=S\setminus(I\cup J),
	\]
	and let $\zeta_{i,r}=\dig_r(\zeta_i)$ denote the column parameter bits whenever they are needed. The common message monomials are
	\begin{align*}
		A&=\prod_{i\in S}a_{i,0},\\
		U_i&=a_{i,1}\prod_{h\in S\setminus\{i\}}a_{h,0}\quad(q_i\le s-2),\\
		R_i&=a_{i,1}A\quad(i\in I),\\
		U_{i,h}&=a_{i,1}a_{h,1}\prod_{\ell\in S\setminus\{i,h\}}a_{\ell,0}
		\quad(i,h\in I,\ i\ne h).
	\end{align*}
	The scalar $Z=\prod_{i\in S}\zeta_{i,0}$ depends only on the column parameters. Monomials involving a second message bit are used only when that bit exists. The symbol $R$ denotes a set of higher-offset factors, while $R_i$ denotes a monomial. In particular, $R_i=a_{i,0}U_i$. These definitions will be reused in all four weight cases below.
	
	\begin{lemma}\label{lem:localcoeff}
		For $s\ge3$, the coefficients needed for total excess degree at most three are, for $i\in I$,
		\begin{align*}
			[T]F_i&=a_{i,0}\zeta_{i,0},\\
			[T^2]F_i&=a_{i,1}\zeta_{i,0}+a_{i,0}\zeta_{i,1},\\
			[T^3]F_i&=a_{i,0}a_{i,1}\zeta_{i,0}+a_{i,0}\zeta_{i,0}\zeta_{i,1},\\
			[T^4]F_i&=a_{i,2}\zeta_{i,0}+a_{i,1}\zeta_{i,1}+a_{i,0}\zeta_{i,2}+a_{i,0}a_{i,1}\zeta_{i,0}\zeta_{i,1},
		\end{align*}
		and, for $j\in J$,
		\[
		[T^2]F_j=a_{j,0}\zeta_{j,0},\qquad [T^4]F_j=a_{j,1}\zeta_{j,0}+a_{j,0}\zeta_{j,1}.
		\]
		For any $s\ge1$ and offset $q_i$, $[T^{2^{q_i}}]F_i=a_{i,0}\zeta_{i,0}$, and subsequent powers are separated by multiples of $2^{q_i}$.
	\end{lemma}
	\begin{proof}
		By Frobenius, $(1+T)^{2^{q_i}x_i\zeta_i}=(1+T^{2^{q_i}})^{x_i\zeta_i}$, so only multiples of $2^{q_i}$ occur. The coefficient at $T^{2^{q_i}}$ is the lowest bit of the product, $a_{i,0}\zeta_{i,0}$. 
		
		For $q_i=0$, use integer arithmetic to compute the product modulo $8$:
		\[
		\begin{aligned}
			x_i\zeta_i &\equiv a_{i,0}\zeta_{i,0}
			+2\bigl(a_{i,1}\zeta_{i,0}+a_{i,0}\zeta_{i,1}\bigr)\\
			&\quad+4\bigl(a_{i,2}\zeta_{i,0}+a_{i,1}\zeta_{i,1}
			+a_{i,0}\zeta_{i,2}\bigr)\pmod 8.
		\end{aligned}
		\]
		The first summand is a bit and produces no carry. The two contributions at weight $2$ are bits, so their sum is reduced modulo $2$ to obtain the second digit. The carry into the third digit is
		\[
		\left\lfloor\frac{a_{i,1}\zeta_{i,0}+a_{i,0}\zeta_{i,1}}{2}\right\rfloor
		=a_{i,0}a_{i,1}\zeta_{i,0}\zeta_{i,1}.
		\]
		Adding this carry to the weight-$4$ contributions modulo $2$ gives the third digit. If $d_0,d_1,d_2$ are the resulting digits, then $(1+T)^{x_i\zeta_i}\equiv(1+d_0T)(1+d_1T^2)(1+d_2T^4)\pmod{T^8}$. Thus the coefficients at degrees $1,2,3,4$ are respectively $d_0,d_1,d_0d_1,d_2$. In simplifying $d_0d_1$, we use $v^2=v$ for every bit. For $q_i=1$, doubling the degrees gives the two coefficients at degrees $2$ and $4$ from the lowest two product digits.
	\end{proof}
	
	\subsection{The maximal weight}
	
	\begin{proposition}\label{prop:top}
		If $w(\vec{k}) = K$, then $\rho_s(\vec{k}) = 1$. 
	\end{proposition}
	\begin{proof}
		In \eqref{eq:coefficient}, the degree-$K$ coefficient is $AZ$ and all lower digits vanish, because $w(\vec k)=K$. Thus every generator is a scalar multiple of $A$. Taking all $\zeta_i=1$ gives the nonzero function $A$, proving $\rho_s(\vec k)=1$.
	\end{proof}
	
	\subsection{One below the maximal weight}
	\begin{proposition}\label{prop:top-2}
		Let $s\ge2$ and $w(\vec k)=K-1$. Then $\rho_s(\vec k)=2$.
	\end{proposition}
	\begin{proof}
		For $s=2$ the only possibility is $\vec k=(1,0)$, covered by Proposition \ref{prop:low}. Suppose $s\ge3$. Since $w(\vec k)$ is odd, $I\ne\varnothing$. The inequality $2^{s-2}<K-1$ shows that only the top digit can contribute. There is exactly one unit of excess degree: it comes either from $b_0T$ in $B_s$ or from a degree-two term of a factor in $I$. By Lemma~\ref{lem:localcoeff},
		\[
		D_{s-1,S,\zeta_S}=Z(b_0A+X)+AY,
		\qquad X=\sum_{i\in I}U_i,\quad
		Y=\sum_{i\in I}\zeta_{i,1}\prod_{h\in S\setminus\{i\}}\zeta_{h,0}.
		\]
		Thus $\mathcal W_s(S)\subseteq\spn_{\F}\{A,b_0A+X\}$. Taking every parameter equal to $1$ gives $b_0A+X$. Changing one parameter in $I$ to $3$ gives $b_0A+X+A$, so both candidate generators belong to the space. They are independent because only the second contains the monomial $b_0A$. Hence the dimension is two.
	\end{proof}
	
	\subsection{Two below the maximal weight}
	\begin{proposition}\label{prop:top-3}
		Let $s\ge 3$ and $w(\vec{k})=K-2$. 
		Then $\rho_s(\vec{k})=k_0+2$. 
	\end{proposition}
	\begin{proof}
		Since $K-2$ is even, $k_0=|I|$ is even. We use the common monomials above; in particular all necessary bits exist because no factor of weight $K$ occurs.
		The candidate generators are
		\[
		g_i=b_0A+\sum_{j\in I\setminus\{i\}}U_j\quad(i\in I)
		\]
		and
		\[
		c_0=b_1A+b_0\sum_{i\in I}U_i+\sum_{i\in I}R_i
		+\sum_{i\in J}U_i+\sum_{\{i,j\}\subseteq I}U_{i,j}.
		\]
		We prove equality with the span of these generators before computing its dimension.
		
		For $s\ge4$, $2^{s-2}<K-2$, so every $D_{j,S,\zeta_S}$ with $j<s-1$ vanishes. For $s=3$, the weight is $2$, so $D_{0,S,\zeta_S}=0$ and $D_{1,S,\zeta_S}=ZA$ by the lowest-degree coefficients in Lemma~\ref{lem:localcoeff}. The expansion of the top digit below applies in both ranges.  
		Write $D(\zeta_S)=D_{s-1,S,\zeta_S}$. 
		Only total excess at most $2$ above $w(\vec k)$ can contribute to $D$. 
		Use Lemma~\ref{lem:localcoeff} and $B_s(T)=1+b_0T+b_1T^2+\cdots$.
		Let $E=\{i\in S:\zeta_{i,0}=0\}$. 
		Expanding these coefficients gives the following exhaustive cases:
		\begin{enumerate}[label=(\roman*),leftmargin=*]
			\item If $E=\varnothing$, then
			\[
			D(\zeta_S)=c_0+\sum_{i\in I} \zeta_{i,1} g_i + 
			\left(\sum_{i\in I} \zeta_{i,1} + \sum_{i\in J} \zeta_{i,1} +
			\sum_{\{i,j\}\subseteq I} \zeta_{i,1} \zeta_{j,1}\right)A.
			\]
			\item If $E=\{i\}$ with $i\in I$, then
			\[
			D(\zeta_S)=\zeta_{i,1}\left(g_i+\sum_{j\in I\setminus\{i\}}\zeta_{j,1}A\right).
			\]
			\item If $E=\{i,j\}\subseteq I$, then $D(\zeta_S)= \zeta_{i,1} \zeta_{j,1} A$.
			\item If $E=\{i\}$ with $i\in J$, then $D(\zeta_S)= \zeta_{i,1} A$.
			\item In every remaining case, $D(\zeta_S)=0$.
		\end{enumerate}
		For the last case, a factor with $\zeta_{i,0}=0$ has no nonzero coefficient before degree $2^{q_i+1}$, so its excess is at least $2^{q_i}$. Every parity pattern not covered by (i)--(iv) has total excess greater than $2$. These cases account for the top digit. For $s=3$, the additional function $D_{1,S,\zeta_S}=ZA$ is also in the same candidate span. This proves
		\[
		\mathcal W_s(S)\subseteq\spn_{\F}\{A,g_i\ (i\in I),c_0\}.
		\]
		Conversely, taking $\zeta_i=1$ for every $i\in S$ yields $D=c_0$. Changing only $\zeta_i$ to $2$ for $i\in I$ yields $D=g_i$. If $k_0\ge2$, changing exactly two level-zero parameters to $2$ yields $D=A$. 
		If $k_0=0$, division of $w(\vec k)=K-2$ by $2$ gives
		\[
		\sum_{\ell=1}^{s-1}2^{\ell-1}k_\ell=2^{s-2}-1,
		\]
		so $k_1$ is odd and $J\ne\varnothing$. Changing one parameter belonging to $J$ to $2$, with all other parameters $1$, again yields $D=A$. Since $k_0$ is even, these exhaust the possibilities. Hence the reverse inclusion holds.
		
		Finally, only $c_0$ contains the monomial $b_1A$, so its coefficient in any dependence is zero. 
		Only $A$ contains the monomial $A$ itself. 
		For a remaining relation $\sum_{i\in I}\lambda_i g_i=0$, the coefficient of $b_0A$ gives $\sum_i\lambda_i=0$, and the coefficient of each $U_j$ gives $\sum_{i\ne j}\lambda_i=0$. 
		Thus every $\lambda_j=0$. This also covers the empty family when $k_0=0$. 
		Therefore $\rho_s(\vec k)=\dim\mathcal W_s(S)=k_0+2$. 
	\end{proof}
	
	\subsection{Three below the maximal weight}
	\begin{proposition}\label{prop:top-4}
		Let $s\ge4$, $K=2^{s-1}$, and $w(\vec k)=K-3$. Then
		\[
		\rho_s(\vec k)=3k_0+2-\varepsilon,\qquad
		\varepsilon=\begin{cases}1,&k_0=1,\ k_1=0,\\0,&\text{otherwise}.\end{cases}
		\]
	\end{proposition}
	\begin{proof}
		Use the common notation and put $n=k_0$. Since $K-3$ is odd, $n$ is odd and positive. Furthermore, $K-3>K/2$ for $s\ge4$, so only the top digit survives in \eqref{eq:coefficient}. Write $D(\zeta)=D_{s-1,S,\zeta_S}$. Only excess degree at most three can contribute.
		
		Introduce the functions
		\begin{align*}
			H_0&=b_0A+\sum_{i\in I} U_i, \\
			H_{i,h}&=b_0A+\sum_{\ell\in I\setminus\{i,h\}}U_\ell \quad(i,h\in I,\ i\ne h), \\
			M_i& = b_0 A+\sum_{h\in I\setminus\{i\}} U_h + R_i
			\quad(i\in I), \\
			L_i&=b_1 A + b_0\sum_{h\in I\setminus\{i\}} U_h 
			+ \sum_{h\in I\setminus\{i\}}R_h+\sum_{j\in J} U_j
			+ \sum_{\{h,\ell\}\subseteq I\setminus\{i\}}U_{h,\ell} + U_i.
		\end{align*}
		Let $C=D(1,\ldots,1)$. Equivalently, its full expression is
		\[
		\begin{split}
			C=[T^3]\ &(1+b_0T+b_1T^2+b_0b_1T^3)\prod_{h\in R}a_{h,0}\\
			&\cdot\prod_{i\in I}(a_{i,0}+a_{i,1}T+a_{i,0}a_{i,1}T^2+a_{i,2}T^3)
			\prod_{j\in J}(a_{j,0}+a_{j,1}T^2).
		\end{split}
		\]
		We will prove that
		\begin{equation}\label{eq:space-first}
			\mathcal W_s(S)=\operatorname{span}_{\mathbb F_2}
			\{A,C,L_i,M_i\ (i\in I)\}+\mathcal H,
		\end{equation}
		where
		\[
		\mathcal H=\operatorname{span}_{\mathbb F_2}
		\bigl(\{H_{i,h}:\{i,h\}\subseteq I\}
		\cup\{H_0:\ k_1>0\}\bigr).
		\]
		The notation $\{H_0:k_1>0\}$ means that $H_0$ is included only if $k_1>0$.
		
		\medskip\noindent\textit{Step 1: an explicit upper bound on the function space.}
		
		Use the coefficients of Lemma~\ref{lem:localcoeff}.
		For $h\in R$ only the natural-degree coefficient $[T^{2^{q_h}}]F_h=a_{h,0}\zeta_{h,0}$ matters: the next possible degree has excess at least four. 
		Thus if any such $\zeta_{h,0}=0$, then $D(\zeta)=0$.
		Assume from now on that $\zeta_{h,0}=1$ for all $h\in R$.
		
		Let $E=\{i\in I\cup J:\zeta_{i,0}=0\}$. The preceding local coefficients give
		an explicit exhaustive expansion. If $E=\varnothing$, it is
		\begin{equation}\label{eq:odd-expansion}
			\begin{split}
				D(\zeta)={}&C+\sum_{i\in I}\zeta_{i,1}(L_i+M_i)
				+\sum_{\{i,h\}\subseteq I}\zeta_{i,1}\zeta_{h,1}H_{i,h}
				+\sum_{j\in J}\zeta_{j,1}H_0\\
				&+\left(\sum_{\{i,h,\ell\}\subseteq I}\zeta_{i,1}\zeta_{h,1}\zeta_{\ell,1}
				+\sum_{i\in I,\,j\in J}\zeta_{i,1}\zeta_{j,1}+\sum_{i\in I}\zeta_{i,2}\right)A.
			\end{split}
		\end{equation}
		If $E=\{i\}$ with $i\in I$, it is
		\begin{equation}\label{eq:one-even}
			\begin{split}
				D(\zeta)={}&\zeta_{i,1}L_i+\zeta_{i,1}\sum_{h\in I\setminus\{i\}}\zeta_{h,1}(H_{i,h}+A)\\
				&+\left(\zeta_{i,2}+\zeta_{i,1}\sum_{j\in J}\zeta_{j,1}
				+\zeta_{i,1}\sum_{\{h,\ell\}\subseteq I\setminus\{i\}}\zeta_{h,1}\zeta_{\ell,1}\right)A.
			\end{split}
		\end{equation}
		The other nonzero possibilities are
		\begin{align*}
			E=\{i,h\}\subseteq I:\quad
			&D(\zeta)=\zeta_{i,1}\zeta_{h,1}\left(H_{i,h}+\sum_{\ell\in I\setminus\{i,h\}}\zeta_{\ell,1} A\right),\\
			E=\{i,h,\ell\}\subseteq I:\quad
			&D(\zeta)=\zeta_{i,1}\zeta_{h,1}\zeta_{\ell,1} A,\\
			E=\{j\}\subseteq J:\quad
			&D(\zeta)=\zeta_{j,1}\left(H_0+\sum_{i\in I}\zeta_{i,1}A\right),\\
			E=\{i,j\},\ i\in I,\ j\in J:\quad
			&D(\zeta)=\zeta_{i,1}\zeta_{j,1}A.
		\end{align*}
		In every remaining case $D(\zeta)=0$. Indeed, an even parameter at offset $q$ increases the possible lowest degree by at least $2^q$. All remaining patterns require excess greater than three. These formulas prove containment of $\mathcal W_s(S)$ in the right-hand side of \eqref{eq:space-first}.
		They follow by choosing local excesses summing to three; in particular, \eqref{eq:odd-expansion} collects the coefficients of the independent parameter bits $\zeta_{i,1},\zeta_{j,1},\zeta_{i,2}$, and \eqref{eq:one-even} collects the degree-two and degree-four terms of the unique even level-zero factor.
		
		\medskip\noindent\textit{Step 2: every candidate generator belongs to the space.}
		
		All unmentioned parameters in the following choices are one. Taking every
		parameter one yields $C$. Changing just one $\zeta_i$ with $i\in I$ to $5$
		yields $C+A$, so $A\in\mathcal W_s(S)$. Changing just $\zeta_i$ to $2$
		yields $L_i$, while changing it to $3$ yields $C+L_i+M_i$; hence $M_i$
		also belongs to the space. Changing two distinct parameters $\zeta_i,
		\zeta_h$ in $I$ to $2$ yields $H_{i,h}$. If $k_1>0$, changing one parameter
		in $J$ to $2$ yields $H_0$. All choices are permitted for $s\ge4$.
		This proves \eqref{eq:space-first}.
		
		\medskip\noindent\textit{Step 3: compute the dimension of $\mathcal H$.}
		
		The monomials $b_0A$ and $U_i$ ($i\in I$) are distinct, so identify their span with $\mathbb F_2^{n+1}$. Since $n$ is odd, every $H_{i,h}$ and $H_0$ lies in the hyperplane
		\[
		\left\{c_0b_0A+\sum_{i\in I}c_iU_i:\ c_0=\sum_{i\in I}c_i\right\}.
		\]
		For $n\ge3$, the $H_{i,h}$ span this hyperplane. To see this, an annihilating linear functional has coefficients $(\lambda_0,(\lambda_i))$ satisfying
		\[
		\lambda_0+\sum_{\ell\in I\setminus\{i,h\}}\lambda_\ell=0
		\quad\text{for all distinct }i,h.
		\]
		Comparing pairs sharing one index gives that all $\lambda_i$ are equal.
		The displayed relation then gives $\lambda_0=\lambda_i$, since $n-2$ is odd. 
		The annihilator is therefore one-dimensional, proving that the span has dimension $n$. In particular, $H_0$ already lies in this span.
		For $n=1$ there are no $H_{i,h}$; $\mathcal H$ has dimension one if $k_1>0$ and dimension zero if $k_1=0$. Thus
		\[
		\dim\mathcal H=n-\varepsilon.
		\]
		
		\medskip\noindent\textit{Step 4: independence of the remaining generators.}
		Consider a dependence
		\[
		\alpha C+\sum_{i\in I}\lambda_iL_i+\sum_{i\in I}\mu_iM_i +\nu A+h=0,\qquad h\in\mathcal H.
		\]
		For any $i\in I$, the monomial $a_{i,2}\prod_{\ell\ne i}a_{\ell,0}$ occurs in $C$ with coefficient one and in none of the other functions, so $\alpha=0$. 
		The coefficient of $b_1A$ then gives $\sum_i\lambda_i=0$. 
		For each $h\in I$, the coefficient of $b_0U_h$ gives $\sum_{i\ne h}\lambda_i=0$, hence $\lambda_h=0$.
		After these coefficients vanish, comparison of the distinct monomials $R_i$ gives $\mu_i=0$. Finally, $A$ does not occur in any element of $\mathcal H$, so $\nu=0$ and $h=0$.
		Consequently the $2n+2$ functions $A,C,L_i,M_i$ are independent modulo $\mathcal H$, and
		\[
		\dim\mathcal W_s(S)=2n+2+\dim\mathcal H=3n+2-\varepsilon.
		\]
		This proves the stated formula.
	\end{proof}
	
	\begin{remark}
		The restriction $s\ge4$ in Proposition~\ref{prop:top-4} is necessary: when $s=3$, the weight-one vector $(1,0,0)$ has coefficient $5$, not $4$.
	\end{remark}
	
	\subsection{A single factor of offset zero}
	For the support with $\vec k=(1,0,\ldots,0)$, write $x$ for its coefficient and $z$ for its normalized column parameter. The local functions specialize to
	\[
	d_{j,z}(b,x)=\dig_j(b+xz)+\dig_j(b),\qquad 0\le j<s,\quad z\in\Z_{2^s}.
	\]
	Let $\mathcal E_s$ be their span on $\Z_{2^s}\times\Z_{2^s}$ and put $r_s=\dim\mathcal E_s=\rho_s(1,0,\ldots,0)$.
	
	\begin{proposition}\label{prop:low}
		For every $s\ge1$,
		\[
		\rho_s(1,0,\ldots,0)=\left\lfloor\frac{2^{s+1}}3\right\rfloor.
		\]
		More precisely, $r_1=1$, $r_2=2$, and
		\[
		r_s=2^{s-1}+r_{s-2}\qquad(s\ge3).
		\]
	\end{proposition}
	\begin{proof}
		For $s=1$, the space is spanned by $a_0=\dig_0(x)$, so $r_1=1$. For $s=2$, it is spanned by $a_0$ and $a_1+b_0a_0$, so $r_2=2$. We now assume $s\ge3$.
		For $j\ge0$ and $1\le r\le2^j$, define
		\begin{equation}\label{eq:G}
			G_{j,r}(b,x)=[T^{2^j}](1+T)^b\left((1+T)^x-1\right)^r.
		\end{equation}
		For a fixed ambient exponent $s$ we use $0\le j<s$. These functions depend only on the residues of $b,x$ needed by the extracted coefficient.
		
		\medskip\noindent\textit{Step 1: replace the parameter values by a triangular basis.}
		Write $A_x(T)=(1+T)^x-1$. Since $A_x(T)$ has zero constant term,
		\begin{equation}\label{eq:newton}
			d_{j,z}=\sum_{r=1}^{2^j}\binom zr G_{j,r}\quad\text{over }\mathbb F_2.
		\end{equation}
		Indeed, expand $(1+T)^{xz}=(1+A_x(T))^z$ and extract the required coefficient after multiplication by $(1+T)^b$. For $z,r=1,\ldots,2^j$, the matrix $(\binom zr\bmod2)$ is lower triangular with diagonal entries one. These parameter values are all available in $\mathbb Z_{2^s}$.
		Thus the $d_{j,z}$ for each $j$ span exactly the $G_{j,r}$ for that $j$.
		
		\medskip\noindent\textit{Step 2: the top-digit subspace.}
		Put $K=2^{s-1}$ and
		\[
		\mathcal T_s=\operatorname{span}_{\mathbb F_2}
		\{G_{s-1,r}:1\le r\le K\}.
		\]
		At $x=1$, $A_1(T)=T$, so
		\[
		G_{s-1,r}(b,1)=\binom b{K-r}\pmod2.
		\]
		The functions $\binom bd\bmod2$ for $0\le d<K$ are precisely the squarefree monomials in $b_0,\ldots,b_{s-2}$: the identity $(1+T)^b=\prod_{h\ge0}(1+b_hT^{2^h})$ proves this directly.
		They are independent. Therefore $\dim\mathcal T_s=K$, and restriction to $x=1$ is injective on $\mathcal T_s$.
		
		\medskip\noindent\textit{Step 3: remove all lower digits on odd $x$.}
		For $j<s-1$, let $L=2^j$ and define
		\[
		Q_{j,r}=G_{j,r}+G_{s-1,K-L+r}\qquad(1\le r\le L).
		\]
		We claim that $Q_{j,r}(b,x)=0$ for every odd $x$. 
		For odd $x$, write $A_x(T)=TV_x(T)$ since 
		\[
		A_x(T)=(1+T)^x-1 =\sum_{r=1}^{x}\binom xr T^r, \qquad
		V_x(T)=\sum_{r=1}^{x}\binom xr T^{r-1}. 
		\]
		Thus, 
		\[
		V_x(T)^L
		=\sum_{r=1}^{x}\left(\binom xr\bmod2\right)T^{L(r-1)}
		\equiv1\pmod{T^L}.
		\]
		Since $K-L$ is divisible by $L=2^j$ for $j<s-1$, then 
		\[
		V_x(T)^{K-L}=\bigl(V_x(T)^L\bigr)^{(K-L)/L}
		\equiv1\pmod{T^L}.
		\]
		It follows that $A_x(T)^{K-L}+T^{K-L}$ is divisible by $T^K$. Hence
		\[
		Q_{j,r}=[T^K](1+T)^b
		\left(A_x(T)^{K-L}+T^{K-L}\right)A_x(T)^r=0,
		\]
		since $r\ge1$. This proves the claim.
		
		\medskip\noindent\textit{Step 4: identify the residual space on even $x$.}
		If $x=2y$, then $A_x(T)=A_y(T^2)$. Also $K-L+r>K/2$, so
		\[
		G_{s-1,K-L+r}(b,2y)=0:
		\]
		its polynomial contains no term of degree at most $K$. Write $b=b_0+2\beta$, with $\beta=\lfloor b/2\rfloor$. For $j\ge1$, extracting even degrees gives
		\begin{align*}
			Q_{j,r}(b,2y)&=G_{j,r}(b,2y)\\
			&=[U^{2^{j-1}}](1+U)^\beta\left((1+U)^y-1\right)^r\\
			&=\begin{cases}
				G_{j-1,r}(\beta,y),&r\le2^{j-1},\\
				0,&r>2^{j-1}.
			\end{cases}
		\end{align*}
		For $j=0$ the function $Q_{0,1}$ is zero also on even $x$, since $A_{2y}(T)$ has no term of degree one.
		
		Let $\mathcal Q_s$ be the span of all $Q_{j,r}$ with
		$0\le j<s-1$ and $1\le r\le2^j$.
		Every function in $\mathcal Q_s$ vanishes on odd $x$.
		On even $x=2y$, writing $\beta=\lfloor b/2\rfloor$, the
		nonzero restrictions of these generators are precisely
		\[
		G_{h,r}(\beta,y),\qquad
		0\le h<s-2,\quad 1\le r\le2^h.
		\]
		Since their defining coefficient extractions involve degrees
		$2^h<2^{s-2}$, the identity
		\[
		(1+U)^{2^{s-2}}=1+U^{2^{s-2}}
		\]
		shows that these functions depend only on $\beta$ and $y$
		modulo $2^{s-2}$. Moreover, as $b$ and even $x$ range over their
		domains, the pair
		\[
		\bigl(\lfloor b/2\rfloor\bmod2^{s-2},
		x/2\bmod2^{s-2}\bigr)
		\]
		ranges over all of $\Z_{2^{s-2}}\times\Z_{2^{s-2}}$.
		By Step~1, the functions $G_{h,r}$ span $\mathcal E_{s-2}$.
		Consequently, restriction to even $x$ and reduction of these
		two arguments induce an isomorphism
		\[
		\mathcal Q_s\cong\mathcal E_{s-2}\qquad(s\ge3).
		\]
		
		Explicitly, the inverse isomorphism sends $f\in\mathcal E_{s-2}$
		to the function
		\[
		(\iota_s f)(b,x)=
		\begin{cases}
			0,&x\text{ odd},\\
			f\bigl(\lfloor b/2\rfloor\bmod2^{s-2},
			x/2\bmod2^{s-2}\bigr),&x\text{ even}.
		\end{cases}
		\]
		The surjectivity of the argument map makes $\iota_s$ injective,
		and the preceding description of the generators shows that
		its image is exactly $\mathcal Q_s$.
		
		\medskip\noindent\textit{Step 5: direct sum and solution of the recurrence.}
		Each lower generator $G_{j,r}$ is the sum of $Q_{j,r}$ and a top generator, so Step 1 gives $\mathcal E_s=\mathcal T_s+\mathcal Q_s$.
		The sum is direct: every function in $\mathcal Q_s$ vanishes at $x=1$, whereas restriction to $x=1$ is injective on $\mathcal T_s$. Consequently
		\[
		\mathcal E_s=\mathcal T_s\oplus\mathcal Q_s,
		\qquad r_s=K+r_{s-2}\quad(s\ge3).
		\]
		Iteration gives
		\[
		r_s=2^{s-1}+2^{s-3}+\cdots=
		\begin{cases}
			(2^{s+1}-1)/3,&s\text{ odd},\\
			(2^{s+1}-2)/3,&s\text{ even}.
		\end{cases}
		\]
		This is $\lfloor2^{s+1}/3\rfloor$, which proves the formula.
	\end{proof}
	
	\section*{Acknowledgment}
	The authors used ChatGPT (OpenAI) and DeepSeek to assist
	with language editing and the exploration of ideas. ChatGPT suggested
	revisions to the exposition in Sections~\ref{sec:rank} and~\ref{sec:app}
	and Appendix~\ref{app:rho}. It also helped draft the reproducibility
	instructions in Appendix~\ref{app:algo} and drafted the supplementary
	verification script \texttt{check\_tables.py}. The authors reviewed and
	revised the AI-assisted content and take responsibility for the
	manuscript and its results.

\end{document}